\pdfoutput=1
\documentclass[11pt]{article} 
\usepackage[utf8]{inputenc} 

\usepackage[margin=1in]{geometry} 
\usepackage{graphicx} 

\usepackage{booktabs} 
\usepackage{array} 
\usepackage{paralist} 
\usepackage{verbatim} 
\usepackage{mathrsfs}
\usepackage{amssymb}
\usepackage{amsthm}
\usepackage{amsmath,amsfonts,amssymb}
\usepackage{esint,appendix}
\usepackage{graphics}
\usepackage{enumerate}
\usepackage{mathtools}
\usepackage{xfrac}
\usepackage{subcaption}
\usepackage{stmaryrd}
 \usepackage{mathabx}

\usepackage[usenames,dvipsnames]{xcolor}
\usepackage[colorlinks=true, pdfstartview=FitV, linkcolor=blue, citecolor=blue, urlcolor=blue]{hyperref}
\usepackage[normalem]{ulem}

\usepackage{tikz}
\usetikzlibrary{calc}
\usepackage{pgf}
\usetikzlibrary{external}
\numberwithin{equation}{section}
\numberwithin{figure}{section}

\newtheorem{theorem}{Theorem}[section]

\newtheorem{corollary}[theorem]{Corollary}
\newtheorem{proposition}[theorem]{Proposition}
\newtheorem{lemma}[theorem]{Lemma}
\theoremstyle{definition}
\newtheorem{definition}[theorem]{Definition}

\newtheorem{remark}[theorem]{Remark}

\newtheorem*{note}{Note}
\newtheorem*{acks}{Acknowledgements}

\newcommand{\abs}[1]{\left|#1\right|}

\newcommand*{\N}{\ensuremath{\mathbb{N}}}

\newcommand*{\Z}{\ensuremath{\mathbb{Z}}}

\newcommand*{\R}{\ensuremath{\mathbb{R}}}

\renewcommand*{\tilde}{\widetilde}

\renewcommand{\P}{\ensuremath{\mathbb{P}}}

\newcommand{\weakto}{\rightharpoonup}

\DeclareSymbolFont{boldoperators}{OT1}{cmr}{bx}{n}
\SetSymbolFont{boldoperators}{bold}{OT1}{cmr}{bx}{n}
\usepackage{accents}

\newcommand{\T}{\mathbb{T}}

\def\XXint#1#2#3{{\setbox0=\hbox{$#1{#2#3}{\int}$}
\vcenter{\hbox{$#2#3$}}\kern-.5\wd0}}

\let\originalleft\left
\let\originalright\right
\renewcommand{\left}{\mathopen{}\mathclose\bgroup\originalleft}
\renewcommand{\right}{\aftergroup\egroup\originalright}

\renewcommand{\>}{\rangle}

\newcommand{\E}{\mathbb{E}}

\renewcommand{\hat}{\widehat}

\makeatletter
\pgfmathdeclarefunction{erf}{1}{%
  \begingroup
    \pgfmathparse{#1 > 0 ? 1 : -1}%
    \edef\sign{\pgfmathresult}%
    \pgfmathparse{abs(#1)}%
    \edef\x{\pgfmathresult}%
    \pgfmathparse{1/(1+0.3275911*\x)}%
    \edef\t{\pgfmathresult}%
    \pgfmathparse{%
      1 - (((((1.061405429*\t -1.453152027)*\t) + 1.421413741)*\t 
      -0.284496736)*\t + 0.254829592)*\t*exp(-(\x*\x))}%
    \edef\y{\pgfmathresult}%
    \pgfmathparse{(\sign)*\y}%
    \pgfmath@smuggleone\pgfmathresult%
  \endgroup
}
\makeatother

\usepackage{titlesec}

\newcommand{\addperiod}[1]{#1.}
\titleformat{\section}
   {\centering\normalfont\Large}{\thesection.}{0.5em}{}
\titleformat*{\subsection}{\bfseries}
\titleformat{\subsubsection}[runin]
  {\normalfont\bfseries}
  {\thesubsubsection.}
  {0.5em}
  {\addperiod}
\titleformat*{\subsubsection}{\normalfont\itshape}
\titleformat*{\paragraph}{\bfseries}
\titleformat*{\subparagraph}{\large\bfseries}

\title{On anomalous diffusion in the Kraichnan model and correlated-in-time variants}

\author{
Keefer Rowan\thanks{Courant Institute of Mathematical Sciences,  New York University.
{\footnotesize \href{mailto:keefer.rowan@cims.nyu.edu}{keefer.rowan@cims.nyu.edu}.}
}
}
\date{}

\begin{document}

\maketitle

\begin{abstract}
    We provide a concise PDE-based proof of anomalous diffusion in the Kraichan model---a stochastic, white-in-time model of passive scalar turbulence. That is, we show an exponential rate of $L^2$ decay in expectation of a passive scalar advected by a certain white-in-time, correlated-in-space, divergence-free Gaussian field, uniform in the initial data and the diffusivity of the passive scalar. Additionally, we provide examples of correlated-in-time versions of the Kraichnan model which fail to exhibit anomalous diffusion despite their (formal) white-in-time limits exhibiting anomalous diffusion. As part of this analysis, we prove that anomalous diffusion of a scalar advected by some flow implies non-uniqueness of the ODE trajectories of that flow.
\end{abstract}


\section{Introduction}

In this paper we consider the Kraichnan model---a model of passive scalar turbulence in which a scalar $\theta$ solves an advection-diffusion equation with advecting flow given by the random white-in-time, correlated-in-space, divergence-free Gaussian field $u$. The flow is specified by two parameters: $\alpha$---which controls the H\"older regularity above the microscale---and $\eta,$ which is the microscopic length scale below which $u$ is smooth. The flow $u$ is given precise specification in Subsection~\ref{subsec.spec-kraichnan-drift}. The SPDE the passive scalar $\theta$ solves is
\begin{equation}
\label{eq.spde-intro}
d\theta_t = \kappa \Delta \theta_t - u \odot \nabla \theta_t,\end{equation}
where the notation of $u \odot \nabla \theta_t$ denotes that we are interpreting the equation in the Stratonovich sense, explained in Subsection~\ref{subsec.spec-kraichnan-spde}. In Section~\ref{s.ad-kraichnan}, we consider the vanishing diffusivity limit $\kappa\to0$ and prove with Theorem~\ref{thm.ad} that this model exhibits \textit{anomalous diffusion}---a uniform-in-$\kappa$ rate of $L^2$ decay of $\theta_t$ in expectation. In Section~\ref{s.nonunique}, we show that for any flow (deterministic or random) anomalous diffusion implies non-uniqueness of positive $L^\infty_tL^2_x$ solutions of  the transport equation associated to the flow as well as nonuniqueness of ODE trajectories in the flow. In Section~\ref{s.no-ad}, we study the necessity of the white-in-time property of the Kraichnan flow by constructing examples of correlated-in-time flows which fail to exhibit anomalous diffusion despite the presence of anomalous diffusion in their formal white-in-time limits. We provide a more precise overview of results in Subsection~\ref{subsec.contribs}, but let us first provide some background on the problems under study.

\subsection{Background and motivation}

The anomalous dissipation of energy in a turbulent fluid is a fundamental experimental fact of turbulence, as emphasized in Frisch's comprehensive account~\cite[Chapter 5]{frisch_turbulence_1995}. Anomalous dissipation refers to persistent dissipation of energy in the vanishing viscosity limit, despite viscosity being the ultimate source of energy dissipation. It is both experimentally and numerically well-observed that this somewhat surprising phenomenon is generic to turbulent fluids~\cite{pearson_measurements_2002, kaneda_energy_2003}. Further, the genericity of anomalous dissipation is a necessary assumption to much of the phenomenological theory of turbulence, in particular it is necessary to the derivation of Kolmogorov's celebrated 4/5-law in K41 theory~\cite{kolmogorov_local_1941, kolmogorov_degeneration_1941, kolmogorov_dissipation_1941}. 

Despite its foundational nature, a satisfying theoretical explanation of anomalous dissipation is still elusive. Giving a complete and rigorous account of anomalous dissipation through the Navier-Stokes equations is currently well beyond the grasp of current techniques. In order to make the problem somewhat more approachable, a simpler but analogous system is often considered. The nonlinear, self-advecting complexity of a turbulent fluid is replaced by passive scalar evolution, in which a field (such as temperature, salinity, or dye concentration) is advected by a flow without acting on the flow. The flow is therefore considered as given, and we are left only to solve a \textit{linear} advection-diffusion equation.

For sufficiently complicated advecting flows---such as when the advecting flow is itself a turbulent fluid---the passive scalar exhibits phenomena analogous to fluid turbulence, which we call \textit{scalar turbulence}, whose heuristic description was first given by Obukhov~\cite{obukhov_structure_1949} and Corrsin~\cite{corrsin_spectrum_1951}. Of particular interest to us, in scalar turbulence we expect generic \textit{anomalous diffusion}---persistence of $L^2$ norm decay in the vanishing viscosity limit. While the linear setting makes the analysis much simpler than for fluid turbulence, it is still very difficult to even give examples of anomalous diffusion. One simple reason to see why is that we need the advecting field to be rough for anomalous diffusion to be even possible: for Lipschitz advecting flows, one can directly control the rate of energy dissipation and prove it vanishes in the $0$ viscosity limit. Thus to construct examples, one needs to provide a rough enough advecting flow while carefully controlling the passive scalar solution. Deterministic examples of anomalous diffusion have only been constructed rather recently, first in~\cite{drivas_anomalous_2022} and~\cite{colombo_anomalous_2022}, which each provide a flow for which one can carefully analyze the associated transport equation and then treat the diffusion perturbatively. Following those examples,~\cite{armstrong_anomalous_2023} gave an example of a flow which one can iteratively homogenize---or renormalize---to show that the effective viscosity on large scales is positive and independent of the molecular viscosity, in the vanishing molecular viscosity limit. 

Prior to the construction of these deterministic examples, there was primarily one model known to the community to exhibit anomalous diffusion. The Kraichnan model \cite{kraichnanSmallScaleStructure1968} is a stochastic model of scalar turbulence in which the flow $u$ is given by a white-in-time, H\"older-continuous-in-space Gaussian field. The Kraichnan model has been subject to substantial investigation in the physics and applied math literature following the foundational paper~\cite{bernard_slow_1998}, in which it was demonstrated to exhibit many of the interesting properties of turbulence---including anomalous diffusion. The bulk of the work done for the Kraichnan model is at the heuristic level and not fully rigorous. For reviews of this literature, see~\cite{falkovichParticlesFieldsFluid2001, kupiainen_nondeterministic_2004, gawedzkiKrzysztofGawedzkiSoluble2008}. 

Thus the Kraichnan model gives a model of scalar turbulence for which anomalous diffusion is generic. It is an appealing alternative to the painstaking construction of particular deterministic examples and suggests a different path for the rigorous study of anomalous diffusion, through suitably chosen random flows. A first step on this path is a rigorous understanding of anomalous diffusion in the Kraichnan model. The primary rigorous reference for the Kraichnan model is the pair of papers~\cite{jan_integration_2002, jan_flows_2004}. These works study a different but related problem to that of anomalous diffusion, instead focusing on the finite time separation of infinitesimally separated particles flowing along the ODE trajectories of the advecting field. Its tools are also quite different from ours, using probabilistic techniques to study Lagrangian particle trajectories. We note also the interesting work~\cite{lototskii_passive_2004}. More recently, the preprint~\cite{zelati_statistically_2023} rigorously studies the Kraichnan model and similar to our analysis essentially uses the closed equation for the equal-time two-point correlation. Their focus is different, instead considering mixing in spatially smooth fields.

\subsection{Contributions of this paper}
\label{subsec.contribs}

This paper provides a rigorous, PDE-centric proof of anomalous diffusion in the Kraichnan model using techniques that are accessible to the fluids community. We particularly focus on how the white-in-time property of the Kraichnan model is needed. 

The Gaussian flow $u^{\eta,\alpha}$ in the Kraichnan model is specified by two parameters, a H\"older exponent $\alpha$ and a small scale cut off $\eta$. If we take the cutoff $\eta$ to be $0$, then the flow is approximately spatially $C^\alpha$. For $\eta>0$, we simply smoothly cutoff length scales below $\eta$, so the flow is spatially smooth but still has the structure of a H\"older continuous flow on scales above $\eta$. The motivation for a small scale cutoff is that an advecting turbulent fluid with positive viscosity should be smooth below the dissipation scale. See Subsection~\ref{subsec.spec-kraichnan-drift} for a precise specification of $u^{\eta,\alpha}$. We show anomalous diffusion for a very broad range of parameters, given by the following precise estimate.

\begin{theorem}
    \label{thm.ad}
    There exists $C(d)$ such that for all $\alpha \in (0,1),$ $\eta \in [0,C^{-1}), \kappa \in (0,C^{-1}),$
    we have for any $\theta_0 \in L^2(\T^d)$ such that $\int \theta_0(x)\ dx =0,$ if $\theta^{\kappa,\eta,\alpha} : [0,\infty) \times \T^d \to \R$ is the random function solving the Kraichnan SPDE 
    \begin{equation}
    \label{eq.spde-thm}
    \begin{cases}
    d\theta^{\kappa,\eta,\alpha}_t = \kappa \Delta \theta^{\kappa,\eta,\alpha}_t - u^{\eta,\alpha} \odot \nabla \theta^{\kappa,\eta,\alpha}_t\\
    \theta^{\kappa,\eta,\alpha}(0,\cdot) = \theta_0,
    \end{cases}
    \end{equation}
    then for all $t >0,$
    \begin{equation} \label{eq.expected-ad}\E \|\theta^{\kappa,\eta,\alpha}_t\|_{L^2(\T^d)}^2 \leq \exp\Big(\frac{C}{1-\alpha} \Big(\frac{\log \kappa}{\log \eta} \lor 1\Big)\Big)e^{-t/C}  \|\theta_0\|_{L^2(\T^d)}^2. \end{equation}
\end{theorem}

We note that the exponential rate of decay is entirely independent of all parameters: only the prefactor varies. We expect anomalous diffusion---estimates on the dissipation of the $L^2$-norm uniform in $\kappa$ as $\kappa \to 0$---only in diagonal limits where $\eta \to 0$ additionally. If $\eta$ stays bounded away from $0$, then the flow is spatially smooth and so cannot exhibit anomalous diffusion. The above estimate then implies for any fixed $\alpha \in (0,1)$ and any $\sigma \in [1,\infty)$, there exists a $C(\sigma, \alpha)<\infty$ such that for any $\eta^\sigma \leq \kappa \leq C^{-1}$ we have that
\[\E \|\theta^{\kappa,\eta,\alpha}_t\|_{L^2(\T^d)}^2 \leq Ce^{-t/C}  \|\theta_0\|_{L^2(\T^d)}^2,\]
that is we get anomalous diffusion in this joint limit. Notably we get anomalous diffusion under an extremely broad class of joint limits $\eta,\kappa \to 0,$ requiring only the very weak condition that the ratio of $\log$'s stay bounded. In particular, one can just take $\eta =0$ and see the estimate is entirely independent of $\kappa.$

The estimate breaks down as $\alpha \to 1$---as expected, since we cannot have anomalous diffusion if the advecting flow is Lipschitz. We note though the interesting dependence on $\alpha$: if the advecting flow is $1-\epsilon$ H\"older, then the estimate suggests we need to wait until $t \approx \epsilon^{-1}$ before the anomalous diffusion is apparent. After this time though, the diffusion continues with an exponential rate independent of $\epsilon$.

Let us quickly sketch how we arrive at the estimate~\eqref{eq.expected-ad}. A simple It\^o calculus computation formally shows, crucially exploiting the white-in-time property of the flow, that the equal-time two point correlation function for $\theta^{\kappa,\eta,\alpha},$ which we denote
\[f^{\kappa,\eta,\alpha}(t,x,y) := \E \theta^{\kappa,\eta,\alpha}(t,x) \theta^{\kappa,\eta,\alpha}(t,y),\]
solves a deterministic linear PDE. We then use the techniques of parabolic PDE together with functional inequalities which are essentially the Caffarelli-Kohn-Nirenberg inequalities to show a decay estimate of $f^{\kappa,\eta,\alpha}$, from which we get the decay of 
\[\int f^{\kappa,\eta,\alpha}(t,x,x)\ dx = \E \int  \theta^{\kappa,\eta,\alpha}(t,x)^2\ dx = \E \|\theta^{\kappa,\eta,\alpha}(t,\cdot)\|_{L^2(\T^d)}^2.\]

Following our demonstration of anomalous diffusion in the Kraichnan model, we turn our attention to the necessity of the white-in-time property of the flow for anomalous diffusion. In particular, we study whether correlated-in-time versions of the Kraichnan model also exhibit anomalous diffusion. While what is perhaps the most natural correlated-in-time version of the Kraichnan model---the model obtained by mollifying the drift field in time---is inaccessible to our current techniques, we construct three different correlated-in-time models that do not exhibit anomalous diffusion for any fixed positive time correlation. We argue that these models can legitimately be called correlated-in-time versions of the Kraichnan model as their (formal) limit as the time correlation goes to $0$ is the Kraichnan model.

Let us sketch one such model. Fix a correlation time $\epsilon>0$, which we will take to $0$ to recover the Kraichnan model. Split time into intervals of width $\epsilon$. On the first interval, we let $u^\epsilon$ be constant in time and spatially, we let it be a random shear flow, randomly chosen to be oriented vertically or horizontally and with profile given by a random $C^{\alpha-}$ Gaussian function with typical magnitude $\epsilon^{-1/2}$. On all other intervals, we let $u^\epsilon$ be iid copies of the first interval. We argue in Subsection~\ref{ss.white-in-time-limits} that $u^\epsilon$ formally converges to a Kraichnan model as $\epsilon\to0,$ but the shear structure also allows us to show that for positive $\epsilon,$ there is no anomalous diffusion. Thus we get the following theorem (for precise statements, see Section~\ref{s.no-ad}).

\begin{theorem}
    \label{thm.no-ad-intro}
    For any $\epsilon>0$, the advection-diffusion equation associated to $u^\epsilon$ does not exhibit anomalous diffusion for any initial data, but in the formal white-in-time limit, the SPDE associated to the limiting flow $u$ exhibits anomalous diffusion for any initial data with some positive probability.
\end{theorem}

In order to show Theorem~\ref{thm.no-ad-intro} and its analogs for the other models, two additional results will be needed. The first is a result that allows us to show that the correlated-in-time models do not exhibit anomalous diffusion. Since these models are also $C^\alpha$ in space, we need some way of showing a $C^\alpha$ flow doesn't generate anomalous diffusion. Given the roughness of the flow, this is less straightforward than it may seem. To this end, we give a proof of the fact that if a flow generates anomalous diffusion, then the ODE trajectories associated to the flow must be non-unique. In particular, we show the following.

\begin{theorem}
    \label{thm.ad-imp-non-unique-ode}
	Suppose that $u \in L^\infty([0,T] \times \T^d)$ with $\nabla \cdot u = 0$, and $u$ exhibits anomalous diffusion. Then there exists a positive final data $\theta_f$ such that the continuity equation
	\begin{equation*}
	\begin{cases}
		\partial_t \theta - \nabla \cdot (u\theta) = 0,\\
		\theta(T,\cdot) = \theta_f,
	\end{cases}
	\end{equation*}
	has non-unique positive solutions weak solutions in $L^\infty([0,T], L^2(\T^d))$. Thus the backward ODE trajectories for $u$, started from time $T$, are non-unique for a positive measure subset of $\T^d$.
\end{theorem}
The above fact about ODE non-uniqueness is implied by the work~\cite{drivasLagrangianFluctuationDissipation2017} using stochastic analysis, but that paper is written in a more applied style and does not state theorems. While a rigorous proof along those lines could be straightforwardly developed, we provide an independent, non-probabilistic proof. The proof the non-uniqueness of the transport equation is similar to that of~\cite[Theorem 3]{drivas_anomalous_2022}. Our proof follows by applying a version of that argument together with splitting into positive and negative parts to prove non-uniqueness of the transport equation for positive solutions. We then conclude by applying Ambrosio's superposition principle~\cite[Theorem 3.2]{ambrosioTransportEquationCauchy2008}, which effectively says that unique ODE trajectories implies unique positive solutions to the associated continuity equation.

As stated above, we will use Theorem~\ref{thm.ad-imp-non-unique-ode} to show that given flows do not exhibit anomalous diffusion. In particular, we combine it with the fact that in 2D, continuous, autonomous, divergence-free vector fields that vanish nowhere have unique ODE trajectories (as shown in~\cite[Theorem 5.1]{silvestre_loss_2013}) to give the following corollary.
\begin{corollary}
    \label{cor.nonvanishing}
    Suppose $u \in C^0(\T^2)$ with $\nabla\cdot u = 0$ and $u$ is nowhere vanishing. Then $u$ does not exhibit anomalous diffusion.
\end{corollary}

The second result needed for Theorem~\ref{thm.no-ad-intro} is anomalous diffusion in an alternative Kraichnan model with a different spatial structure to the flow than the conventional model---that is it is specified by a different covariance matrix than usual. This Kraichnan model appears as the formal white-in-time limit of the flow in Theorem~\ref{thm.no-ad-intro} and is built on shear flows as opposed to generic divergence-free flows. Its precise specification is in Subsection~\ref{ss.shear-kraichnan}. Showing anomalous diffusion in this model follows identically the argument given in Section~\ref{s.ad-kraichnan} for the usual Kraichnan model, except different functional inequalities analogous to the Caffarelli-Kohn-Nirenberg inequalities are needed, namely Proposition~\ref{prop.weighted-poincare-nonradial} and Proposition~\ref{prop.est-for-nash-nonradial}.

\section{Kraichnan model}
\label{s.kraichnan-spec}

The flow specified by the Kraichnan model is a random Gaussian vector field. The flow is taken to be \textit{white-in-time}, which is to say distinct time slices are independent. This makes the flow a.s.\ distributional (as opposed to a classical function) in time, and so we we need a solution theory of SPDE to handle the associated advection-diffusion equation. In this section, we first give a precise specification of the flow and then of the associated SPDE. The white-in-time assumption is in some sense deeply unphysical, as a realistic flow might have very short time correlations, but certainly would not be uncorrelated-in-time. The reason the white-in-time assumption is introduced---despite the technical difficulties associated with the SPDE it induces and despite its unphysical nature---is that it allows us to write down a closed (deterministic) equation for the equal-time, two-point correlation function (as well as equations for the higher order correlation functions, though that is not used here), as is noted in~\cite{falkovichParticlesFieldsFluid2001}. We introduce the equation for the correlation function in this section. This equation is a degenerate parabolic equation that we can apply PDE estimates to in order to compute the decay of the expected $L^2$ norm of passive scalars. As such, we can translate all the analytic difficulty of the problem from analyzing a complex stochastic equation to analyzing a relatively simple degenerate diffusion equation. This certainly is not the case if one were to take the flow to be correlated-in-time, thus this section demonstrates how this argument depends essentially on the white-in-time property of the flow.

\subsection{Specification of the random drift}
\label{subsec.spec-kraichnan-drift}

We take the specification of the Kraichnan model similarly to~\cite{falkovichParticlesFieldsFluid2001}. We take $u$ to be the stationary zero-mean Gaussian field on $\T^d$ with covariance given as
\[\E u^{\eta,\alpha}_i(s,x) u^{\eta,\alpha}_j(t,y) = D^{\eta,\alpha}_{ij}(x-y) \delta(t-s),\]
where we give the Fourier transform of $D^{\eta,\alpha}_{ij}$,
\[
\hat D^{\eta,\alpha}_{ij}(k) = 
\begin{cases}
 \left(\delta_{ij} - \frac{k_i k_j}{|k|^2} \right) |k|^{-(d+2\alpha)} \rho(\eta |k|) & k \ne 0 \\
 0 & k =0,
\end{cases}\]
where $\rho$ is a smooth decreasing function such that $\rho(0) = 1$ and $\rho$ vanishes faster than any polynomial at $\infty.$ The paradigmatic examples of $\rho$ are
\[\rho(t) = e^{-t}\ \text{or}\ \rho(t) = e^{-t^2},\]
though of course that are many other admissible choices. We note that we let all constants freely depend on $\rho$. The term $\delta_{ij} - \frac{k_ik_j}{|k|^2}$ in $\hat D^{\eta,\alpha}_{ij}(k)$ is introduced in order to ensure $u$ is divergence free (in particular one can compute that $\E \nabla \cdot u^{\eta,\alpha}(s,x) \nabla \cdot u^{\eta,\alpha}(t,y) = 0$, so $\nabla \cdot u^{\eta,\alpha} = 0$ a.s.).

The two free parameters, $\alpha$ and $\eta$, both control regularity. The regime we will be interested in will be $\alpha$ fixed and $\eta \to 0$. In this regime, we see that the cutoff only effects very large $k$, so $\eta$ governs the \textit{spatial small scale regularity} of $u$. In particular, it ensures that for any positive $\eta$ that the field is spatially smooth. Physically, one should view $\eta$ as the dissipation scale induced by the viscosity of the advecting flow. Heuristically, $u$ is spatially smooth on length scales well below $\eta.$  

The other parameter $\alpha$ then controls the \textit{large scale regularity} of $u$. In particular, one should imagine that on length scales much bigger than $\eta$, $u$ ``looks'' $C^{\alpha-}$.\footnote{We use the notation $C^{\alpha-} := \bigcap_{\beta<\alpha} C^\beta$} In particular, Kolmogorov-Chentsov theorem gives that for $\eta=0$, $u$ is a.s.\ spatially $C^{\alpha-}$.

\subsection{Specification of the SPDE}
\label{subsec.spec-kraichnan-spde}

We now turn our attention to the equation solved by the scalar 
advected by the stochastic drift field. Recall the usual drift-diffusion equation, say for a smooth deterministic flow $v$,
\[\partial_t \theta - \kappa \Delta \theta + v \cdot \nabla \theta =0.\]
We would like to put our stochastic field $u$ in for $v$ and have that be the equation for $\theta$. Unfortunately, even though this equation is completely linear, the interpretation of such an equation with a white-in-time drift field is non-trivial. The whiteness-in-time elevates the equation from a more friendly \textit{random PDE} to a fully fledged \textit{stochastic PDE}. This brings with it a fairly large layer of technicalities just to do the usually simple existence theory. An important wrinkle is that the noise is acting \textit{multiplicatively}, in that instead of there being a white-in-time stochastic additive forcing, the stochasticity is in the drift, which acts multiplicatively against $\theta$. The presence of multiplicative noise creates a distinction between the It\^o interpretation of the equation and Stratonovich interpretation.

Taking a diversion into SDEs to illustrate the point in a simpler setting, let us consider the SDE 
\begin{equation}
\label{eq.formal-sde}
\dot X^j(t) = \sum_{i=1}^n f^j_i(X(t)) \xi^i(t),\end{equation}
where the $\xi^i$ are standard white noises (the ``time derivatives'' of standard Brownian motions). This is the setting of multiplicative noise in the SDE setting. The equation is purposefully written informally to illustrate the interpretative difficulties. The usual way that this sort of equation interpreted mathematically is in the It\^o sense, usually written
\[dX^j_t = \sum_{i=1}^n f^j_i(X_t) dW^i_t,\]
and a solution is such that for each $t$
\[X^j_t = X^j_0 + \sum_{i=1}^n\int_0^t  f^j_i(X_s) dW^i_s,\]
where the integral is an It\^o integral.

Another reasonable way to try to interpret~\eqref{eq.formal-sde} would be through mollification. One could replace the distributions $\xi^i$ with $\xi^i_\epsilon := \eta_\epsilon * \xi^i$ where $\eta_\epsilon$ is a standard family of mollifiers. Then for each $\epsilon>0$, we get a well defined random ODE which we could for the stochastic process $X^j_\epsilon(t)$. Then we take the limit as $\epsilon \to 0$. 

One may expect that this process would recover the usual It\^o solution to the SDE. \textit{This turns out to be false in general.} What the solutions $X_\epsilon$ to the mollified equations converge to is generically the solution to the \textit{Stratonovich SDE}. This SDE is usually written as
\[dX_t^j =\sum_{i=1}^n f_i^j(X_t) \circ dW_t^i,\]
and solutions $X_t^j$ are such that for every $t,$
\[X_t^j = X_0^j +  \sum_{i=1}^n \int_0^t f_i^j(X_s) \circ dW_t^i,\]
where the integral is interpreted as the Stratonovich integral. The distinction between the It\^o and Stratonovich integrals effectively amounts to a difference in the convention for computing the ``Riemann sums'' for these integrals, with the It\^o convention corresponding to a left Riemman sum and the Stratonovich to a midpoint Riemann sum.

Since the It\^o SDE is not recovered by mollifying the noise, for many physical models the natural SDE model is the Stratonovich SDE, as one is generically ambivalent between very short time correlations and white-in-time correlations. On the other hand, it turns out that It\^o calculus is often mathematically more convenient than Stratonovich calculus. Fortunately, Stratonovich SDEs can be phrased as equivalent It\^o SDEs (and vice versa). The It\^o to Stratonovich conversion gives that the following SDEs are equivalent
\[dX^j_t = b^j(X_t) dt + \sum_{i=1}^n f^j_i(X_t) \circ dW^i_t \]
and
\begin{equation}
    \label{eq.strat-to-ito}
    dX^j_t = \left(b^j(X_t) + \frac{1}{2} \sum_{k,i=1}^n \partial_k f^j_i(X_t) f^k_i(X_t)\right)dt + \sum_{i=1}^n f^j_i(X_t) dW^i_t,
\end{equation}
where the first is a Stratonovich SDE and the second an It\^o SDE. In particular, we note they differ by a deterministic drift term. For a more complete discussion, see \cite{evansIntroductionStochasticDifferential2012}.

Returning to the Kraichnan SPDE, the presence of multiplicative noise forces us to choose a convention---It\^o or Stratonovich---in defining the equation. The above discussion motivated that the Stratonovich convention is the natural one for this equation. The Stratonovich convention is also what is universally used in the literature on the Kraichnan model. Using the notation of SDEs, we write the Kraichnan SPDE as
\begin{equation}
\label{eq.strat-spde}
d\theta^{\kappa,\eta,\alpha}_t = \kappa \Delta \theta^{\kappa,\eta,\alpha}_t - u^{\eta,\alpha} \odot \
\nabla \theta^{\kappa,\eta,\alpha}_t.
\end{equation}
The equivalent It\^o SPDE is the following
\begin{equation}
\label{eq.ito-spde}
d\theta^{\kappa,\eta,\alpha}_t = \left(\tfrac{1}{2}D(0) +\kappa I\right): \nabla^2 \theta^{\kappa,\eta,\alpha}_t - u^{\eta,\alpha} \cdot \nabla \theta^{\kappa,\eta,\alpha}_t.
\end{equation}
We see it involves a correction term similar to the SDE case. For a derivation of the Stratonovich to It\^o conversion for the Kraichnan model, see~\cite{dunlapQuenchedLocalLimit2022} or~\cite[Section 2.3]{galeati_convergence_2020}.\footnote{As the latter reference makes particularly clear, one way of computing this correction is to view the SPDE as an infinite system of SDEs driven by iid Brownian motions and use the Stratonovich-to-It\^o correction described above for SDEs to compute the correction for the SPDE. This is straightforward but laborious and unenlightening, so we refrain from repeating the argument here.}

The existence theory for this sort of SPDE is technical but well developed. The details are not relevant to the present study. Unique solutions can be found in the generalized solution sense due to Kunita \cite{kunitaFirstOrderStochastic1984}. This solution theory is developed for this exact model (on $\R^d$ instead of $\T^d$, but everything follows similarly) in \cite{dunlapQuenchedLocalLimit2022}. An alternative $L^2$-based solution theory is developed in~\cite{galeati_convergence_2020}.

What we will need for our study of the Kraichnan model is the PDE for the equal-time two-point correlation
\[f^{\kappa,\eta,\alpha}(t,x,y) := \E \theta^{\kappa,\eta,\alpha}(t,x) \theta^{\kappa,\eta,\alpha}(t,y).\]
We provide the derivation with formal It\^o calculus---suppressing temporarily the superscripts.
\begin{align*}
    \partial_t f(t,x,y) &= \E d (\theta(t,x) \theta(t,y))
    \\&= \E d\theta_t(x) \theta_t(y) + \E \theta_t(x) d\theta_t(y) + \E \<d\theta_t(x), d\theta_t(y)\>
    \\&= \left(\tfrac{1}{2}D(0) + \kappa \right):(\nabla_x^2 + \nabla_y^2) \E \theta(t,x) \theta(t,y) +  \<u_t(x), u_t(y)\>_{ij}\partial_{x_i} \partial_{y_j} \E  \theta(t,x)  \theta(t,y)
    \\&= \left(\tfrac{1}{2}D(0) + \kappa \right):(\nabla_x^2 + \nabla_y^2) f(t,x,y) +   D_{ij}(x-y)\partial_{x_i} \partial_{y_j} f(t,x,y).
\end{align*}
This computation is rigorously justified in \cite{dunlapQuenchedLocalLimit2022}.

We note the equation for $\theta$ is translation invariant, so if we start with random initial data with translation invariant law $\theta_0(x) \stackrel{d}{=} \theta_0(x+v)$ for any $v \in \R^d,$ then $\theta$ will remain translation invariant in law, in particular, $f(t,x,y)$ will invariant under the translation $f(t,x,y) = f(t,x+v,y+v)$ for any $v \in \R^d$. Thus
\[f(t,x,y) = f(t,x-y,0) =: g(t,x-y).\]
Then we note that $g$ solves
\[
\partial_t g = ( 2\kappa I +  D(0)- D(r)): \nabla^2 g = \nabla \cdot (2\kappa I +D(0) - D(r)) \nabla g,\]
where we use that $D$ is divergence-free.

We summarize the consequence of the above discussion as the following proposition.

\begin{proposition}
    Let $u$ be the white-in-time Gaussian field specified above. Then the Stratonovich SPDE
    \[d\theta_t^{\kappa,\eta,\alpha} = \kappa \Delta \theta_t^{\kappa,\eta,\alpha} - u^{\eta,\alpha} \odot \nabla \theta_t^{\kappa,\eta,\alpha},\]
    is formally equivalent to the It\^o SPDE
    \[d\theta^{\kappa,\eta,\alpha}_t = \left(\tfrac{1}{2}D^{\eta,\alpha}(0) +\kappa I\right): \nabla^2 \theta^{\kappa,\eta,\alpha}_t - u^{\eta,\alpha} \cdot \nabla \theta^{\kappa,\eta,\alpha}_t.\]
    This SPDE has unique solutions in the sense of Kunita (or in the sense of energy solutions). Letting 
    \[f^{\kappa,\eta,\alpha}(t,x,y) := \E \theta^{\kappa,\eta,\alpha}(t,x) \theta^{\kappa,\eta,\alpha}(t,y),\]
    then $f$ solves the PDE 
    \[\partial_t f^{\kappa,\eta,\alpha}(t,x,y) =\left(\tfrac{1}{2}D^{\eta,\alpha}(0) + \kappa \right):(\nabla_x^2 + \nabla_y^2) f^{\kappa,\eta,\alpha}(t,x,y) +   D^{\eta,\alpha}_{ij}(x-y)\partial_{x_i} \partial_{y_j} f(t,x,y).\]
    In the case that $\theta_0$ is translation invariant in law, then so is $\theta^{\kappa,\eta,\alpha}$, and thus $f^{\kappa,\eta,\alpha}$ is translation invariant, so
    \[\E \theta^{\kappa,\eta,\alpha}(t,x) \theta^{\kappa,\eta,\alpha}(t,y) = f^{\kappa,\eta,\alpha}(t,x,y) = g^{\kappa,\eta,\alpha}(t,x-y)\]
    where $g^{\kappa,\eta,\alpha}$ solves
    \[\partial_t g^{\kappa,\eta,\alpha} = \nabla \cdot a^{\kappa,\eta,\alpha}(x) \nabla g^{\kappa,\eta,\alpha}\]
    with 
    \begin{equation}
    \label{eq.a-def}
    a^{\kappa,\eta,\alpha}(r) := 2 \kappa I +D^{\eta,\alpha}(0) - D^{\eta,\alpha}(r).
    \end{equation}
\end{proposition}

\section{Anomalous diffusion for the Kraichnan model}
\label{s.ad-kraichnan}

In this section, we keep track of explicit dependence of all constants on $\alpha,\eta,\kappa$; abstract analytic constants $C$ depend only on the dimension $d$ (and implicitly the cutoff function $\rho$). We are interested in anomalous diffusion, so we are interested in the $L^2$ diffusion in the limit as $\kappa \to 0$ for fixed $\alpha$. We will have to take a simultaneous limit as $\eta \to 0$ though in order to get the diffusion anomaly. This is because the advecting field is very smooth on scales well below $\eta$, and we need roughness on small scales in order to get the diffusion anomaly. Thus we are interested in the diagonal limits $\eta,\kappa \to 0$. Anomalous diffusion under the proper class of diagonal limits will be a consequence of the explicit bounds in terms of $\eta, \kappa.$

It is worth noting here that this argument works only in the case that $\alpha <1,$ that is the spatial structure of the field must be less regular than Lipschitz. Thus we see that the anomalous diffusion in the Kraichnan model is generated both by the white-in-time nature of the field \textit{and its spatial roughness.} A white-in-time but spatially smooth field would not work.

Throughout, we assume that $\theta^{\kappa,\eta,\alpha}$ has initial data that is translation invariant in law. Note that 
\[g^{\kappa,\eta,\alpha}(t,0) = \E \theta^{\kappa,\eta,\alpha}(t,0)^2 = \E \theta^{\kappa,\eta,\alpha}(t,x)^2 = \frac{1}{C}\E \int \theta^{\kappa,\eta,\alpha}(t,x)^2\ dx = \frac{1}{C}\E \|\theta^{\kappa,\eta,\alpha}(t,\cdot)\|_{L^2(\T^d)}^2.\]
Thus we can show (expected) $L^2$ diffusion of $\theta^{\kappa,\eta,\alpha}$ by showing time decay of $g^{\kappa,\eta,\alpha}(t,0),$ which we now take as our primary goal. Recall that
\[\partial_t g^{\kappa,\eta,\alpha} =  \nabla \cdot a^{\kappa,\eta,\alpha}(x) \nabla g^{\kappa,\eta,\alpha},\]
so $g$ solves a divergence-form diffusion equation. Thus to show the desired decay, we need control on the diffusion matrix $a^{\kappa,\eta,\alpha}.$ The following proposition gives the desired control. We defer the computationally intensive proof---which is similar to computations that have appeared previously in the literature, such as in~\cite{eyink_existence_1996-tex-fixed}---to Appendix~\ref{a.diffusion-matrix-computation}. Note that throughout we make the notational identification of $\T^d = [-\pi,\pi]^d.$

\begin{proposition}
 \label{prop.diffusion-matrix-bound}
    There exists $c(d) >0$ such that for any $\alpha \in (0,1), \beta \in [\alpha,1], \eta \in [0,c), \kappa \geq 0, w \in \R^d,$
        \begin{equation} 
        \label{eq.matrix-lower-bound-beta}
        w \cdot a^{\kappa,\eta,\alpha}(x) w \geq c (1 \land \eta^{2\beta (\alpha-1)} \kappa^{1-\beta}) |x|^{2\beta} |w|^2,\end{equation}
        in particular,
        \begin{equation}
        \label{eq.matrix-lower-bound-square}
        w \cdot a^{\kappa,\eta,\alpha}(x) w \geq c |x|^2 |w|^2.
        \end{equation}
\end{proposition}

\begin{remark}
    The above expression~\eqref{eq.matrix-lower-bound-beta} is a bit complicated as we keep track of explicit constants in $\kappa,\eta,\alpha$ and introduce a free parameter $\beta$. This constant dependence turns out to be reasonably simple in the final result (Theorem~\ref{thm.ad}), but for ease of reading, one can keep in mind the somewhat natural case that $\beta = \alpha, \kappa \approx \eta^{2\alpha}.$ We see in this case we get the bound that 
    \[w \cdot a^{\eta^{2\alpha},\eta,\alpha} w \geq c |x|^{2\alpha} |w|^2.\]
\end{remark}

These bounds~\eqref{eq.matrix-lower-bound-beta} and~\eqref{eq.matrix-lower-bound-square} are all the control we will need to prove the desired decay. Our goal is to show a uniform rate of time decay of $g^{\kappa,\eta,\alpha}(t,0)$. Since, for fixed $\kappa>0$, the matrix $a^{\kappa,\eta,\alpha}(x)$ is uniformly elliptic, classical parabolic theory gives that the solution $g^{\kappa,\eta,\alpha}$ is continuous. As such, it suffices to show the $L^\infty$ decay of $g^{\kappa,\eta,\alpha}$. Our proof will proceed directly along the usual lines: first we show the $L^2$ decay of $g^{\kappa,\eta,\alpha}$ and then we show an $L^2$ bound for the fundamental solution. These together then imply an $L^\infty$ bound on $g^{\kappa,\eta,\alpha}.$

The proofs will also be essentially the usual arguments of parabolic theory but with different functional inequalities than those used in the proof for uniformly elliptic matrices. The only parabolic estimate is the energy identity,
\[\frac{d}{dt} \frac{1}{2} \|g^{\kappa,\eta,\alpha}\|_{L^2(\T^d)}^2 = -\int \nabla g^{\kappa,\eta,\alpha} \cdot a^{\kappa,\eta,\alpha}(x) \nabla g^{\kappa,\eta,\alpha}\ dx.\]
Normally, one uses the Poincar\'e inequality for the $L^2$ decay and the original argument of Nash uses a special case of the Gagliardo-Nirenberg interpolation inequalities to get the $L^2$ control of the fundamental solution.

Instead of having access to the $L^2$ norm of the gradient of the solution---as is the case for uniformly elliptic diffusion matrices---we only have access to a version weighted by a power of $|x|$, 
\[\int |x|^{2\beta} |\nabla g^{\kappa,\eta,\alpha}|^2\ dx.\]
Thus instead of using the inequalities suggested above, we use weighted versions, which are effectively the Caffarelli-Kohn-Nirenberg inequalities~\cite{caffarelliFirstOrderInterpolation1984}.

In particular we need a version of the Poincar\'e inequality to do the $L^2$ estimate. This is provided by the following. Note the proposition is stated for $0$-mean functions on the hypercube and as such holds \textit{a fortiori} for $0$-mean functions on the torus.

\begin{proposition}
    \label{prop.weighted-poincare-base}
    Let $d \geq 2$, $D:= [-\pi,\pi]^d,$ and $g : D \to \R$ such that $\int g\ dx = 0$. Then there exists $C(d) < \infty$ such that
    \[\|g\|_{L^2(D)} \leq C \||x| \nabla g\|_{L^2(D)}.\]
\end{proposition}

We also will need the following weighted interpolation inequality for the estimate on the fundamental solution.

\begin{proposition}
    \label{prop.est-for-nash}
    Let $d \geq 2$, $\beta \in (0,1)$, and $g : \T^d \to \R,$ then there exists $C(d)< \infty$ such that 
    \[\|g\|_{L^2(\T^d)} \leq C (\| |x|^\beta \nabla g\|_{L^2(\T^d)}^a + \|g\|_{L^2(\T^d)}^a) \|g\|_{L^1(\T^d)}^{1-a}.\]
    with
    \[a = \frac{d}{d+2-2\beta} \in (0,1).\]
\end{proposition}

Inequalities of this form are most often stated for compactly supported functions on $\R^d.$ We provide proofs of the versions we need in Appendix~\ref{a.ineq-proofs}. The arguments are included for completeness as no references for the needed cases ($0$-mean and periodic) could be found, but no claim to novelty is made.

With these inequalities in hand, together with the bounds~\eqref{eq.matrix-lower-bound-square} and ~\eqref{eq.matrix-lower-bound-beta}, we are ready to prove the $L^\infty$ decay of $g^{\kappa,\nu}.$ We proceed by first showing exponential-in-time $L^2$ decay and then controlling the $L^\infty$ norm by the $L^2$ norm. First, the $L^2$ bound, which starts with the usual energy estimate.

\begin{proposition}[Energy estimate]
    \label{prop.weighted-energy}
    Suppose that $g$ solves 
    \begin{equation}
    \label{eq.degenerate-diff-eq}
    \partial_t g - \nabla \cdot a^{\kappa,\eta,\alpha} \nabla g = 0
    \end{equation}
    with $\int g(0,x)\ dx =0$ and $a^{\kappa,\eta,\alpha}$ as defined in~\eqref{eq.a-def}. Then there exists $c(d) >0$ such that for all $\eta<c, \kappa \geq 0, \alpha \in (0,1)$
    we have the decay
    \[\|g(t,\cdot)\|_{L^2(\T^d)} \leq e^{-c t} \|g(0,\cdot)\|_{L^2(\T^d)}.\]
\end{proposition}

\begin{proof}
    Note that the equation is mean-preserving, so that $\int g(t,x)\ dx =0 $ for all $t\geq 0$. We then compute
    \begin{align*}
        \frac{d}{dt}  \|g\|_{L^2(\T^d)}^2 = -2\int \nabla g \cdot a^{\kappa,\eta,\alpha} \nabla g \leq -c\int |x|^{2} |\nabla g|^2 \leq -c \|g\|_{L^2(\T^d)}^2,
    \end{align*}
    where we use the bound~\eqref{eq.matrix-lower-bound-square} for the first inequality and Proposition~\ref{prop.weighted-poincare-base} for the last inequality. Then the Gr\"onwall inequality gives the proposition.
\end{proof}

Thus we have exponential $L^2$ decay of $g^{\kappa,\eta,\alpha}$ with a uniform rate, but as discussed above, we really want $L^\infty$ decay of $g^{\kappa,\eta,\alpha}$. To go from $L^2$ to $L^\infty$, we get $L^2$ control of the fundamental solution, which is provided by the following proposition.

\begin{proposition}[Nash estimate]
    \label{prop.weighted-nash}
    Let $\Phi(t,x,y)$ be the fundamental solution to the equation~\eqref{eq.degenerate-diff-eq} started at $y$, i.e.\ 
    \[\begin{cases}
        \partial_t \Phi - \nabla_x \cdot a^{\kappa,\eta,\alpha}(x) \nabla_x \Phi = 0\\
        \Phi(0,\cdot,y) = \delta_y.
    \end{cases}\]
    then there exists $C(d)$ such that for all $\eta < C^{-1},\kappa > 0, \alpha \in (0,1), \beta \in [\alpha,1),$ and for all $x,y$,
    \[\|\Phi(t,x,\cdot)\|_{L^2_y(\T^d)} + \|\Phi(t,\cdot,y)\|_{L^2_x(\T^d)} \leq \Big(\frac{C}{1-\beta}\Big)^{\frac{d}{4(1-\beta)}}\Big( \eta^{\frac{d\beta (1-\alpha)}{2(1-\beta)}} \kappa^{-d/4} +1\Big)t^{-{\frac{d}{4(1-\beta)}}} + C^{\frac{d}{4(1-\beta)}},\]
\end{proposition}

\begin{proof}
    Note that since $a^{\kappa,\eta,\alpha}$ is symmetric, $\Phi$ is symmetric in $x,y$, 
    \[\Phi(t,x,y) = \Phi(t,y,x).\]
    As such we just prove the estimate on $\|\Phi(t,\cdot,y)\|_{L^2_x(\T^d)}.$

    Fix some $y \in \T^d$, let $\phi(t,x) := \Phi(t,x,y).$ Note that $\int \phi(t,x)\ dx = \|\phi(t,\cdot)\|_{L^1(\T^d)} = 1$. We then write the usual energy identity
    \[\frac{d}{dt} \|\phi\|_{L^2(\T^d)}^2 = - 2 \int \nabla \phi \cdot a^{\kappa,\eta,\alpha} \nabla \phi \leq - c  (1 \land \eta^{2\beta (\alpha-1)} \kappa^{1-\beta})\||x|^\beta \nabla \phi\|_{L^2(\T^d)}^2,\]
    where we use~\eqref{eq.matrix-lower-bound-beta}. Then by Proposition~\ref{prop.est-for-nash}
    \[\|\phi\|_{L^2(\T^d)} \leq C (\||x|^\beta \nabla \phi\|_{L^2(\T^d)}^a + \|\phi\|_{L^2(\T^d)}^a) \|\phi\|_{L^1(\T^d)}^{1-a} =C \||x|^\beta \nabla \phi\|_{L^2(\T^d)}^a + C\|\phi\|_{L^2(\T^d)}^a,\]
    with
    \[a = \frac{d}{d+2-2\beta}.\]
    Thus
    \[\||x|^\beta \nabla \phi\|_{L^2(\T^d)}^a \geq c \|\phi\|_{L^2(\T^d)} - \|\phi\|_{L^2(\T^d)}^a.\]
    Thus since $a<1$, there is some $C(d)< \infty$ such that if $\|\phi\|_{L^2(\T^d)} \geq C$, then
    \[c\|\phi\|_{L^2(\T^d)} - \|\phi\|_{L^2(\T^d)}^a \geq c \|\phi\|_{L^2(\T^d)} - \frac{c}{2} \|\phi\|_{L^2(\T^d)} \geq c\|\phi\|_{L^2(\T^d)},\]
    so
    \[\||x|^\beta \nabla \phi\|_{L^2(\T^d)} \geq  c \|\phi\|_{L^2(\T^d)}^{1/a}\]
    and so, for $\|\phi\|_{L^2(\T^d)}^2 \geq C$,
    \[\frac{d}{dt} \|\phi\|_{L^2(\T^d)}^2 \leq - c  (1 \land \eta^{2\beta (\alpha-1)} \kappa^{1-\beta})\|\phi\|_{L^2(\T^d)}^{2/a} = - c (1 \land \eta^{2\beta (\alpha-1)} \kappa^{1-\beta})(\|\phi\|_{L^2(\T^d)}^2)^{1 + \frac{2-2\beta}{d}}.\]
    Let
    \[g(t) := (\|\phi\|_{L^2(\T^d)}^2)^{- \frac{2-2\beta}{d}}.\]
    Note then that for $g(t) \leq c(d)$ (as we need $\|\phi\|_{L^2(\T^d)} \geq C$), we have that
    \[\frac{d}{dt} g(t) = - \frac{2-2\beta}{d} (\|\phi\|_{L^2(\T^d)}^2)^{-1 - \frac{2-2\beta}{d}}\frac{d}{dt} \|\phi\|_{L^2(\T^d)}^2 \geq c  (1-\beta) (1 \land \eta^{2\beta (\alpha-1)} \kappa^{1-\beta}).\]
    Note also that $\lim_{t \to 0} \|\phi\|_{L^2(\T^d)}^2(t) = \infty,$
    so $\lim_{t \to 0} g(t) = 0.$
    Then we have that
    \[g(t) \geq  \min(c(1-\beta) (1 \land \eta^{2\beta (\alpha-1)} \kappa^{1-\beta})t,c),\]
    where we get the min as the differential inequality stops being valid for $g(t) \geq c$. Thus we have that
    \[\|\phi\|_{L^2(\T^d)} \leq \Big(\frac{C}{1-\beta}\Big)^{\frac{d}{4(1-\beta)}}\Big( \eta^{\frac{d\beta (1-\alpha)}{2(1-\beta)}} \kappa^{-d/4} +1\Big)t^{-{\frac{d}{4(1-\beta)}}} + C^{\frac{d}{4(1-\beta)}},\]
    allowing us to conclude.
\end{proof}

We can then combine the above two results to give $L^\infty$ decay of solutions to the degenerate parabolic equation.

\begin{proposition}[$L^\infty$ decay of solutions]
    \label{prop.l-infty-decay}
    Suppose that $g$ solves 
    \[\partial_t g - \nabla \cdot a^{\kappa,\eta,\alpha} \nabla g = 0\]
    with $\int g(0,x)\ dx =0.$ Then there exists $C(d) >0$ such that for all $\alpha \in (0,1)$, $\eta,\kappa \leq C^{-1}, \kappa> 0,$
    we have for all $t>0,$
    \[\|g(t,\cdot)\|_{C^0(\T^d)} \leq \exp\Big(\frac{C}{1-\alpha} \Big(\frac{\log \kappa}{\log \eta} \lor 1\Big)\Big)e^{-t/C} \|g(0,\cdot)\|_{L^\infty(\T^d)},\]
    and also for all $t\geq \frac{C}{1-\alpha} \Big(\frac{\log \kappa}{\log \eta} \lor 1)$
    \[\|g(t,\cdot)\|_{C^0(\T^d)} \leq\exp\Big(\frac{C}{1-\alpha} \Big(\frac{\log \kappa}{\log \eta} \lor 1\Big)\Big)e^{-t/C} \|g(0,\cdot)\|_{L^2(\T^d)}.\]
\end{proposition}

\begin{proof}
    We have that the matrix $a^{\kappa,\eta,\alpha}$ is uniformly elliptic for any $\kappa>0$, so the solution $g$ is continuous for any $t>0$. As such it suffices to bound the $L^\infty$ norm.

    The first inequality is a direct consequence of the second inequality and the maximum principle. Let us prove then the second inequality. Using Proposition~\ref{prop.weighted-energy} and Proposition~\ref{prop.weighted-nash}, we have for any $0 < \tau < t,$
     \begin{align*}    
    \|g(t,x)\|_{L^\infty(\T^d)} &= \Big\|\int g(t-\tau,y) \Phi(\tau,x,y)\ dy\Big\|_{L^\infty_x(\T^d)}
    \\&\leq \|\Phi(\tau,x,y)\|_{L^\infty_x L^2_y} \|g(t-\tau,\cdot)\|_{L^2(\T^d)} 
    \\&\leq \Big(\Big(\frac{C}{1-\beta}\Big)^{\frac{d}{4(1-\beta)}}\Big( \eta^{\frac{d\beta (1-\alpha)}{2(1-\beta)}} \kappa^{-d/4} +1\Big)e^{\tau/C}\tau^{-{\frac{d}{4(1-\beta)}}} + C^{\frac{d}{4(1-\beta)}}e^{\tau/C}\Big) e^{-t/C} \|g(0,\cdot)\|_{L^2(\T^d)}.
    \end{align*}
    Taking $\tau = \frac{1}{1-\beta}$, we get that
    \begin{align*}    
    \|g(t,x)\|_{L^\infty(\T^d)} &\leq C^{\frac{1}{1-\beta}}\Big(\eta^{\frac{d\beta (1-\alpha)}{2(1-\beta)}} \kappa^{-d/4} +1 \Big) e^{-t/C} \|g(0,\cdot)\|_{L^2(\T^d)}
    \\&=C^{\frac{1}{1-\beta}}\Big(\eta^{\frac{d(1-\alpha)}{2(1-\beta)}}  (\eta^{2(1-\alpha)}\kappa)^{-d/4} +1 \Big) e^{-t/C} \|g(0,\cdot)\|_{L^2(\T^d)}.
    \end{align*}
    Optimizing while respecting the constraint that $\beta \geq \alpha$, we take
    \[\beta = \frac{\log \kappa}{\log \eta^{2(1-\alpha)} + \log \kappa} \lor \alpha,\]
    and so if $\kappa \geq \eta^{2 \alpha}$, we take $\beta = \alpha$ and so get for $t \geq \frac{1}{1-\alpha},$
    \[ \|g(t,x)\|_{L^\infty(\T^d)} \leq C^{\frac{1}{1-\alpha}}\Big(\eta^{\frac{d\alpha}{2}} \kappa^{-d/4} +1 \Big) e^{-t/C} \|g(0,\cdot)\|_{L^2(\T^d)} \leq  C^{\frac{1}{1-\alpha}} e^{-t/C} \|g(0,\cdot)\|_{L^2(\T^d)}.\]
    Otherwise, if $\kappa \leq \eta^{2\alpha},$ we get that for $t \geq 1 + \frac{1}{2(1-\alpha)} \frac{\log \kappa}{\log \eta}$
    \[ \|g(t,x)\|_{L^\infty(\T^d)} \leq C^{1 + \frac{1}{1-\alpha}\frac{\log \kappa}{\log \eta}} e^{-t/C} \|g(0,\cdot)\|_{L^2(\T^d)}.\]
    We see then these can be combined to give that for $t \geq \frac{C}{1-\alpha} \Big(\frac{\log \kappa}{\log \eta} \lor 1)$
    \[\|g(t,x)\|_{L^\infty(\T^d)} \leq \exp\Big(\frac{C}{1-\alpha} \Big(\frac{\log \kappa}{\log \eta} \lor 1\Big)\Big)e^{-t/C} \|g(0,\cdot)\|_{L^2(\T^d)},\]
    thus we conclude.
\end{proof}

This proposition easily implies Theorem~\ref{thm.ad}.

\begin{proof}[Proof of Theorem~\ref{thm.ad}]
    All the propositions above are for initial data that are translation invariant in law. Letting $\theta_0 \in L^2(\T^d)$ deterministic, by adding an independent random uniform translation to $\theta_0$, we can get a random initial data $\tilde \theta_0$ that is translation invariant in law, and further
    \[\E \tilde \theta_0^2(y)= \frac{1}{|\T^d|}\int_{\T^d} \theta_0^2(x+y)\ dx = c \|\theta_0\|_{L^2(\T^d)}^2.\]

    Consider $g^{\kappa,\eta,\alpha}$ associated to the initial data $\tilde \theta_0$. Note that
    \[\int g^{\kappa,\eta,\alpha}(0,x)\ dx = \E \tilde \theta_0(0) \int \tilde \theta_0(x)\ dx = 0.\]
    Note also that
    \[\abs{g^{\kappa,\eta,\alpha}(0,x)} = \abs{\E \tilde \theta_0(0) \tilde \theta_0(x)} \leq \tfrac{1}{2} \E \tilde \theta_0(0)^2 + \tfrac{1}{2} \E \tilde \theta_0(x)^2  = c \|\theta_0\|_{L^2(\T^d)}^2.\]
    Let $\theta^{\kappa,\eta,\alpha}_t$ be the solution the Kraichnan SPDE~\ref{eq.spde-thm} with initial data $\theta_0$ and $\tilde \theta^{\kappa,\eta,\alpha}_t$ the solution with initial data $\tilde \theta_0$. Note that $\tilde \theta^{\kappa,\eta,\alpha}_t$ is just an independent uniform random translation of $\theta^{\kappa,\eta,\alpha}_t$
    Thus, using Proposition~\ref{prop.l-infty-decay}, we have that
    \begin{align*}\E \| \theta^{\kappa,\eta,\alpha}_t\|_{L^2(\T^d)}^2 
    &= \E \|\tilde \theta^{\kappa,\eta,\alpha}_t\|_{L^2(\T^d)}^2 
    \\&= C g^{\kappa,\eta,\alpha}(t,0) \leq C\|g^{\kappa,\eta,\alpha}(t,\cdot)\|_{C^0(\T^d)} 
    \\&\leq \exp\Big(\frac{C}{1-\alpha} \Big(\frac{\log \kappa}{\log \eta} \lor 1\Big)\Big)e^{-t/C} \|g(0,\cdot)\|_{L^\infty(\T^d)} 
    \\&\leq  \exp\Big(\frac{C}{1-\alpha} \Big(\frac{\log \kappa}{\log \eta} \lor 1\Big)\Big)e^{-t/C} \|\theta_0\|_{L^2(\T^d)}^2,
    \end{align*}
    as desired.
\end{proof}

\section{Anomalous diffusion implies non-uniqueness of backward ODE trajectories}
\label{s.nonunique}

In this section we give a new proof that anomalous diffusion implies non-uniqueness of ODE trajectories of the underlying field. This fact is of some independent interest, but it is in particular used here to show that some of the correlated-in-time models constructed below do not exhibit anomalous diffusion.

Let us first give the a very broad definition of what we mean for a vector field to exhibit anomalous diffusion.

\begin{definition}
	Let $u \in L^\infty([0,T] \times \T^d)$, $\nabla \cdot u =0$, $\theta_0 \in L^2(\T^d)$, and define $\theta^\kappa : [0,T] \times \T^d$ to be the unique solution to
	\begin{equation*}
	\begin{cases}
		\partial_t \theta^\kappa - \kappa \Delta \theta^\kappa - \nabla \cdot (u \theta^\kappa) = 0\\
		\theta^\kappa(0,\cdot) = \theta_0.
	\end{cases}
	\end{equation*}
	We say that $u$ \textit{exhibits anomalous diffusion with initial data $\theta_0$} if 
	\[\liminf_{\kappa\to0} \|\theta^\kappa\|_{L^2(\T^d)}(T) < \|\theta_0\|_{L^2(\T^d)}.\]
	We say that $u$ \textit{exhibits anomalous diffusion} if it exhibits anomalous diffusion for some initial data.
\end{definition}

We note that under this definition, Theorem~\ref{thm.ad} gives that the Kraichnan model, for any (non-constant) initial data exhibits anomalous diffusion with positive probability for sufficiently large $T$. With this definition, we restate Theorem~\ref{thm.ad-imp-non-unique-ode}.
\begingroup
\def\thetheorem{\ref{thm.ad-imp-non-unique-ode}}
\begin{theorem}
Suppose that $u \in L^\infty([0,T] \times \T^d)$ with $\nabla \cdot u = 0$, and $u$ exhibits anomalous diffusion. Then there exists a positive final data $\theta_f$ such that the continuity equation
	\begin{equation*}
	\begin{cases}
		\partial_t \theta - \nabla \cdot (u\theta) = 0,\\
		\theta(T,\cdot) = \theta_f,
	\end{cases}
	\end{equation*}
	has non-unique positive solutions weak solutions in $L^\infty([0,T], L^2(\T^d))$. Thus the backward ODE trajectories for $u$, started from time $T$, are non-unique for a positive measure subset of $\T^d$.
\end{theorem}
\addtocounter{theorem}{-1}
\endgroup

In the above theorem, we reference weak solutions to the continuity equation with final data. For reference, we define weak solutions with initial or final data.

\begin{definition}
    Let $u \in L^\infty([0,T] \times \T^d).$ We say that $\theta \in L^\infty([0,T], L^2(\T^d))$ is a weak solution to the continuity equation
    \begin{equation*}
        \begin{cases}
            \partial_t \theta - \nabla \cdot (u\theta) =0\\
            \theta(0,\cdot) = \theta_0,
        \end{cases}
    \end{equation*}
    if for every $\phi \in C_c^\infty([0,T) \times \T^d)$, we have that 
    \[\int -\partial_t \phi \theta + u \cdot \nabla \phi \theta\ dxdt - \int \theta_0(x) \phi(0,x)\ dx =0.\]
    Similarly, $\theta$ is a weak solution to the continuity equation
        \begin{equation*}
        \begin{cases}
            \partial_t \theta - \nabla \cdot (u\theta) =0\\
            \theta(T,\cdot) = \theta_f,
        \end{cases}
    \end{equation*}
    if for every $\phi \in C_c^\infty((0,T] \times \T^d)$, we have that 
    \[\int -\partial_t \phi \theta + u \cdot \nabla \phi \theta\ dxdt + \int \theta_0(x) \phi(T,x)\ dx =0.\]
\end{definition}

We will want to use the following properties of weak solution in proof of Theorem~\ref{thm.ad-imp-non-unique-ode}.

\begin{lemma}
    \label{lem.properties-of-weak-solutions-cont}
    Let $u \in L^\infty([0,T] \times \T^d)$ and let $\theta \in L^\infty_t L^2_x$ be a weak solution to the initial value problem
    \begin{equation*}
    \begin{cases}
        \partial_t \theta - \nabla\cdot (u\theta) = 0\\
        \theta(0,\cdot) = \theta_0.
    \end{cases}
    \end{equation*}
    Then 
    \begin{enumerate}
        \item We have the representation
        \[\theta = \theta_0 + \int_0^t \nabla \cdot (u\theta),\]
        with the integral interpreted as Bochner integral over $H^{-1}$.
        \item $\theta \in C([0,T]; L^2_w(\T^d)),$ after modifying on a measure zero set.
        \item For any $\phi \in C^\infty(\R \times \T^d)$ and any $0 \leq s<r \leq T$, we have that
        \[\int_s^r \int -\partial_t \phi \theta + u \cdot \nabla \phi \theta\ dxdt + \int \phi(r,x) \theta(r,x) - \phi(s,x) \theta(s,x)\ dx,\]
        when using the continuous-in-time representation of $\theta$.
    \end{enumerate}
\end{lemma}

A proof using standard tools is provided for the reader's convenience in Appendix~\ref{a.weak-continuity}. With this in hand, we are prepared to prove Theorem~\ref{thm.ad-imp-non-unique-ode}.

\begin{proof}[Proof of Theorem~\ref{thm.ad-imp-non-unique-ode}]
	The implication from non-unique positive weak solutions in $L^\infty([0,T], L^2(\T^d))$ to non-unique backward ODE trajectories is effectively direct from Ambrosio's superposition principle~\cite[Theorem 3.2]{ambrosioTransportEquationCauchy2008}. For a careful application of the superposition principle to show the desired non-uniqueness, see~\cite[Proof of Theorem 1.3]{bruePositiveSolutionsTransport2021}.

    In order to show the non-uniqueness to the transport equation, we will first construct a strictly diffusive solution using the assumed anomalous diffusion. Then we will take the final data from that solution and use that to construct a new solution to the continuity equation with that final data as the vanishing viscosity limit of a \textit{backward heat equation}. Then this solution will have norm nonincreasing \textit{backward-in-time,} i.e.\ norm nondecreasing forward-in-time. While the original solution constructed from the anomalous diffusion will have a norm which decreases forward-in-time, thus showing they are two distinct solutions to the continuity equation with the same final data. Some care will need to be taken in giving positive solutions, which is necessary for the application of Ambrosio's superposition principle. This will be done by considering the positive and negative parts of the initial data to the drift-diffusion equation separately.

    Let $\theta_0$ be the initial data for which $u$ exhibits anomalous diffusion. Before proceeding, for technical convenience, we extend $u$ to $[0,2T] \times \T^d$ by $u|_{[T,2T]} = 0$ and let $\theta^\kappa$ be the unique solution to
    \begin{equation*}
	\begin{cases}
		\partial_t \theta^\kappa - \kappa \Delta \theta^\kappa - \nabla \cdot (u \theta^\kappa) = 0,\\
		\theta^\kappa(0,\cdot) = \theta_0,
	\end{cases}
	\end{equation*}
    on $[0,2T] \times \T^d$.

    Then, by the definition of exhibiting anomalous diffusion, we have that
	\[\liminf_{\kappa\to0} \|\theta^\kappa\|_{L^2_x}(T) < \|\theta_0\|_{L^2}.\]
	Let $\theta_0^+, \theta_0^-$, denote the positive in negative parts of $\theta_0$, so that
	\[\theta_0 = \theta_0^+ - \theta_0^-.\]
	Let $\theta^{\kappa,+}, \theta^{\kappa,-}$ be the unique solutions to the above drift-diffusion equation with $\kappa$ diffusivity but with initial data $\theta_0^+$ and $\theta_0^-$ respectively, so that 
	\[\theta^\kappa = \theta^{\kappa,+} - \theta^{\kappa,-}.\]
	Using weak compactness and taking subsequences, let $\kappa_j \to 0$ such that
	\[\theta^{\kappa_j,+} \stackrel{L^\infty_tL^2_x}{\weakto} \theta^+;\qquad \theta^{\kappa_j,-} \stackrel{L^\infty_tL^2_x}{\weakto} \theta^-;\qquad \|\theta^{\kappa_j}\|_{L^2_x}^2(T) \to E < \|\theta_0\|_{L^2}^2.\]
	Define 
	\[\theta := \theta^+ - \theta^- \in L^\infty_t L^2_x.\]
	Note that, by monotonicity of energy for $\kappa>0$ and by norms only dropping in limits, we have that, for a.e.\ $t \geq T$, 
	\[\|\theta\|_{L^2_x}^2(t) \leq E.\]
    Using the $L^2_w$ continuous representation provided by Lemma~\ref{lem.properties-of-weak-solutions-cont}, we then get the above inequality for every $t \geq T$.\footnote{This is the only place we use the extension to $[0,2T]$.} In particular,
    \[\|\theta\|_{L^2_x}^2(T) \leq E.\]
    
	Note that $\theta^+, \theta^-$ are weak solutions to the continuity equation
	\[\partial_t \gamma - \nabla \cdot (u\gamma) =0\]
	with initial data $\theta_0^+, \theta_0^-$ respectively. Further, since positivity is preserved under weak limits, we have that $\theta^+, \theta^-$ are positive solutions to the continuity equation.

    Then, by Lemma~\ref{lem.properties-of-weak-solutions-cont} and taking the continuous representations, we have that $\theta^+, \theta^-$ also solve the above continuity equation on $[0,T] \times \T^d$ with final data $\theta^+_f := \theta^+(T,\cdot), \theta^-_f := \theta^-(T,\cdot)$ respectively.

    For each $\kappa>0$, define $\tilde \theta^{\kappa,+}$ as the unique solution to the final value problem
	\begin{equation*}
		\begin{cases}
		\partial_t \tilde \theta^{\kappa,+} + \kappa \Delta \tilde \theta^{\kappa,+} - \nabla \cdot (u \tilde \theta^{\kappa,+}) = 0,\\
		\tilde \theta^{\kappa,+}(T,\cdot) = \theta^+_f,
		\end{cases}
	\end{equation*}
	and similarly define $\tilde \theta^{\kappa,-}$.

    Using weak compactness and taking a subsequential limit, we get for some $\kappa_j \to0$
	\[\tilde \theta^{\kappa_j,+} \stackrel{L^\infty_tL^2_x}{\weakto}\tilde \theta^+,\]
	so that $\tilde \theta^+ \in L^\infty([0,T],L^2(\T^d))$ is such that $\theta^+$ solves
	\[\partial_t \tilde \theta^+ + \nabla \cdot (u \tilde \theta^\kappa) = 0\]
	on $(0,T]$ with final data $\theta^+_f,$ and we analogously get $\tilde \theta^-$.

    Thus we see that $\theta^+, \theta^-$ solve the same final value problems for the continuity equation as $\tilde \theta^+, \tilde \theta^-$ respectively and, since weak limits preserve positivity, they are both positive solutions to the continuity equations. Thus to conclude, it suffices to show that  $(\theta^+, \theta^-) \ne (\tilde \theta^+, \tilde \theta^-)$. In particular, we show that
	\[\theta := \theta^+ - \theta^- \ne \tilde \theta^+ - \tilde \theta^- =: \tilde \theta.\]

    Note, letting 
	\[\tilde \theta^\kappa := \tilde \theta^{\kappa, +} - \tilde \theta^{\kappa,-},\]
	we get that
	\begin{equation*}
	\begin{cases}
		\partial_t \tilde \theta^\kappa  + \kappa \Delta \tilde \theta^\kappa - \nabla \cdot (u \tilde \theta^\kappa) =0\\
		\tilde \theta^\kappa(T,\cdot) = \theta_f := \theta^+_f - \theta^-_f = \theta(T,\cdot). 
	\end{cases}
	\end{equation*}
	We further have that
	\[\tilde \theta^{\kappa_j} \stackrel{L^\infty_t L^2_x}{\weakto} \tilde \theta.\]
	Then by the monotonicity of energy in the equation for $\tilde \theta^\kappa$, we have that for any $t \in [0,T]$ and any $\kappa$
	\[\|\tilde \theta^\kappa\|_{L^2_x}^2(t) \leq \|\tilde \theta^\kappa\|_{L^2_x}(T)^2 = \|\theta_f\|_{L^2_x}^2(T) \leq  E.\]
	Then by the fact that weak limits can only drop norms, we have that
	\[\|\tilde \theta\|^2_{L^\infty_t L^2_x} \leq E < \|\theta_0\|_{L^2}^2.\]

    On the other hand $\|\theta\|^2_{L^\infty_t L^2_x} = \|\theta_0\|_{L^2}^2$. To see this, it's first clear by weak limits only decreasing norms that $\|\theta\|^2_{L^\infty_t L^2_x}  \leq \|\theta_0\|_{L^2}^2$. On the other hand, we have that $\theta$ is continuous in $L^2_w$. So taking $t_j \to 0$ such that
    \[\|\theta(t_j)\|_{L^2} \leq \|\theta\|_{L^\infty_t L^2_x},\]
    we have that $\theta(t_j) \stackrel{L^2}{\weakto} \theta_0$, so 
    \[\|\theta_0\|_{L^2} \leq \limsup \|\theta(t_j)\|_{L^2} \leq \|\theta\|_{L^\infty_t L^2_x},\]
    so we get equality.

    Thus $\theta \ne \tilde \theta$, so in particular $(\theta^+, \theta^-) \ne (\tilde \theta^+, \tilde \theta^-),$ thus the continuity equation has some positive final data such that there are non-unique positive solutions.
\end{proof}

\section{Correlated-in-time models and lack of anomalous diffusion}

\label{s.no-ad}

We now turn our attention to the construction and analysis of correlated-in-time variants of the Kraichnan model. As we saw above, the white-in-time correlation of the Kraichnan model makes its analysis much simpler as we can get closed equations for the multi-point equal-time correlations. No such tool will work once we introduce non-trivial time correlations in the advecting field $u$. 

Let us outline one heuristic for understanding the meaningful distinction in diffusive behavior between the white-in-time nature of the Kraichnan model and correlated-in-time models. It was shown in~\cite{drivasLagrangianFluctuationDissipation2017} that anomalous diffusion happens if and only if the advecting field exhibits spontaneous stochasticity, the property that the limiting behavior of the SDE drift-diffusion in the vanishing noise limit remains non-deterministic. In other words, particle trajectories perturbed by arbitrarily small noise will grow to have a finite variance in finite time independent of the size of the perturbing noise. This in turn---heuristically at least---is related to finite-time separation of nearby particles being transported by the advecting field, independent of the size of the initial separation.

For the Kraichnan model, since on each time slice the drift field is entirely independent of the previously seen drift field, we can think of two nearby particles as receiving correlated kicks, where the correlation depends only on the separation of the particles and is independent of their histories. As such, the only thing that determines the rate at which nearby particles separate is the rate of decay of the correlation of the kicks these particles receive. This is seen in the above proof, as we needed $D(0) - D(x)$ to grow fast enough. Thus the Kraichnan model will always give anomalous diffusion, provided the spatial field is rough enough, since the roughness of the field is precisely related to the rate of decay of the correlations.

On the other hand, in a correlated-in-time model no such analysis is available to us. The actual spatial and temporal structure of the advecting field come to play a much greater role. For example, the advecting field could be ``fluid-like'', in the sense that it is (approximately) self-advecting, with the different modes transporting each other. Alternatively, it could be (locally) frozen, taken to be piecewise constant in time. The distinction between these choices cannot be seen in the white-in-time limit, but they can lead to meaningfully different dynamics of the advected passive scalar as well as advected particles. In particular, having a frozen-in-time field will mean that sweeping effects, the presence of slowly varying and large magnitude modes, will cause an advected particle to rapidly pass over fast oscillating modes, causing averaging of the fast modes. On the other hand, this effect should not be present in a model that has the fast modes being advected by the slow modes.

These rough heuristics suggest that the white-in-time nature of the Kraichnan model is likely vital to the generation of the diffusion anomaly. In further demonstration of this idea, we provide three examples of correlated-in-time models for which there is no anomalous diffusion despite their (formal) white-in-time limit having a diffusion anomaly.

\subsection{Correlated-in-time models and their formal white-in-time limits}

\label{ss.white-in-time-limits}

Before proceeding to the specific models, we will first explain the general class of models we will be studying and in what sense they are appropriately called correlated-in-time version of a white-in-time model.

The models we will be considering will be piecewise constant in time fields, scaled so that they formally converge to a white noise. In particular, for each model we will fix some mean-zero distribution over spatially varying fields and let $u_j$ be iid fields from that distribution. Then we let
\[u^\epsilon(t,x) := \epsilon^{-1/2}u_{\lceil t/\epsilon \rceil}(x).\]
Note that these are scaled so that their time integral from $0$ to $t$ has variance proportional to $t$, just as a Brownian motion does. Further, one can check that the multitime covariance of $u^\epsilon$ approximates a $\delta$ as $\epsilon \to 0$.

We will show that these correlated-in-time models we construct do not exhibit anomalous diffusion. In particular, we will show for fixed $\epsilon>0$, the $L^2(\T^d)$ norm becomes constant in the limit $\kappa \to 0$ of vanishing molecular diffusivity. We will also be showing the associated white-in-time model will exhibit anomalous diffusion, as they will be variants of the Kraichnan model studied above.

For these examples to be compelling, we need to argue that white-in-time Kraichnan model variants really correspond to the $\epsilon \to 0$ limit of the correlated-in-time models. The convergence of an evolution equation driven by a correlated-in-time version of a white noise to the SDE driven by the actual white noise is known as a Wong-Zakai theorem, after the original investigation of this limit by Wong and Zakai~\cite{wongConvergenceOrdinaryIntegrals1965,wongRiemannStieltjesApproximationsStochastic1969}. It is worth recalling that, as we noted above, we generically expect the limiting SDE to be driven by Stratonovich noise rather than It\^o noise. This fits our needs well as the Kraichnan model is stated with Stratonovich noise.

The most common version of a Wong-Zakai theorem is for the simplest colorings of the noise, e.g.\ by taking the correlated-in-time version of the noise to be the mollification of the white noise or taking it to be piecewise constant Gaussians. We will be working with somewhat more general correlated-in-time models, where the noise is taken to be piecewise constant in time, but is not necessarily Gaussian.

The statement of a Wong-Zakai-type theorem for SDEs where the noise is being generated by a non-Gaussian distribution is given in~\cite{breuillardRandomWalksRough2009}, where they call it a Wong-Zakai-Donsker-type theorem, in that we are simultaneously getting the Donsker-type convergence of a non-Gaussian random walk to a Brownian motion and the Wong-Zakai convergence of the solutions to the stochastic evolution equations. In particular, that result applies for finite dimensional SDEs and says that the white-in-time limit (taken in the same way we are taking ours) converges to the Stratonovich SDE driven by a Gaussian noise that is white-in-time and has the same spatial covariance as the distribution generating the correlated-in-time noise.

What we'd really want is a Wong-Zakai-Donsker-type result for the Kraichnan SPDE. While there are some results on Wong-Zakai theorems for SPDEs, a result that also includes the Donsker invariance part, allowing for non-Gaussian correlated-in-time noises, seems not currently known, and it's investigation is certainly beyond the scope of our current study. As such, let us take the result for finite dimensional SDEs to be sufficient motivation to say that Kraichnan model with the same spatial covariance is the appropriate white-in-time model of the correlated-in-time models we construct here.

Lastly, before moving on to the construction of the models, let us note that one may consider more complicated schemes for introducing time correlations to the Kraichnan model. Of particular interest is allowing different scales to have different time correlations, so that highly oscillatory modes have shorter time correlations. Models of this sort are considered in~\cite{chaves_lagrangian_2003}. It seems likely that if the large frequency modes have short enough correlations in time, then the model exhibits anomalous diffusion, but a rigorous analysis of these models is beyond our current reach.

\subsection{First example: a spatially smooth model}

Here, and in the following sections, we will disregard the parameter $\eta$ that was considered above, taking it always to be $0$ for notational and conceptual simplicity.

Note that for any fixed $\alpha \in (0,1),$ the Kraichnan drift, as specified in subsection~\ref{subsec.spec-kraichnan-drift} with $\eta=0$, can be written using its sine and cosine series as
\[u^0(t,x) = \sum_j c_j f_j(x) dW_t^j,\]
where the $f_j$ are smooth, $\sup_j \|f_j\|_{L^\infty(\T^d)} < \infty, \sum_j c_j^2 =1,$ and the $W_t^j$ are standard, independent Brownian motions.\footnote{This is effectively just the Fourier representation of $u,$ and the $f_j$ just scaled products of $\sin$ and $\cos$ with the scale factor taken so that we can make $c_j^2$ sum to $1$. The fact that $u$ is taken to be real, that we want  $c_j^2$ sum to $1$, and the vector indices floating around all make it unwieldy to write this representation explicitly, but it's not hard to see it exists.} Note then that $u^0$ has the covariance
\[\E u^0(t,x) u^0(s,y) = \delta(t-s) \sum_j c_j^2 f_j(x) f_j(y).\]

Define the random field $\mu$ such that for each $j,$ 
\[\P(\mu = f_j) = c_j^2.\]
Since $\sum_j c_j^2 = 1$ and the $f_j$ are distinct, this completely determines $\mu$. Let $Z$ be an independent standard normal. Let, for each $k$, let the random field $u_k$ be the an independent and identically distributed copy of the product $Z\mu$. Note then that
\[\E u_k(x) = 0\]
and
\[\E u_k(x) u_k(y) = \sum_j c_j^2 f_j(x) f_j(y) \E Z^2 = \sum_j c_j^2 f_j(x) f_j(y),\]
which is the same as the spatial covariance of $u^0$. Then we define the correlated in time drift field $u^\epsilon$ by
\[u^\epsilon(t,x) := \epsilon^{-1/2}u_{\lceil t/\epsilon \rceil}(x).\]
From the discussion above, we see that the usual ($\eta =0$) Kraichnan model is the (formal) $\epsilon \to 0$ limit of the model given by the correlated-in-time drift field $u^\epsilon$. As such, we have from Theorem~\ref{thm.ad} that the white-in-time limit of $u^\epsilon$ exhibits anomalous diffusion.

On the other hand, consider the correlated in time model for fixed positive $\epsilon$ and up to a fixed time $T$. Then we see that, for each realization of $u^\epsilon$, it is piecewise constant on the $T/\epsilon$ intervals of length $\epsilon$. On each of these intervals, $u^\epsilon$ is just one of the $f_j$, as such it is spatially smooth. Thus $u^\epsilon$, on each realization separately, is a piecewise constant in time on a finite set of disjoint intervals and spatially smooth on each of these intervals. In particular, for each realization, $u^\epsilon \in L^\infty_t W^{1,\infty}_x$. It is thus easy to see that, for fixed $\epsilon >0$, there is no anomalous diffusion in the $\kappa \to 0$ limit. This discussion can be summarized in the following proposition.

\begin{theorem}
    Let $u^\epsilon$ be the random field defined as above. Then, for each $\epsilon>0$ and for each realization, $u^\epsilon$ is spatially smooth. In particular, it is in $L^\infty_t W^{1,\infty}_x$. As such, no realization of $u^\epsilon$ for positive $\epsilon$ exhibits anomalous diffusion. On the other hand, $u^\epsilon$ has the (formal) white-in-time limit $u^0$, the Kraichnan drift field, which does exhibit anomalous diffusion.
\end{theorem}

\subsection{Anomalous diffusion in an alternative Kraichnan model}

\label{ss.shear-kraichnan}

For the next two examples, the white-in-time model will not quite be the usual Kraichnan model which we analyzed above. The examples are built on shears, as such the spatial structure of $u$ will be that of the sum of two random shears, 
\[u_s(x,y) := f(x) e_y + g(y) e_x,\]
where $f,g$ are then centered random $C^\alpha$ functions $\T \to \R$. We take the covariance of $f$ to be
\[\E f(x)f(x') = D_f(x-x')\]
with
\[\hat D_f(k) =  \frac{1_{k \ne 0}}{|k|^{1+2\alpha}}.\]
$g$ is then taken to be independent and identically distributed to $f$. Then note that
\[D_s(x-w,y-z) := \E u^i(x,y) u^j(w,z) = \begin{pmatrix} D_f(y-z) & 0 \\ 0 & D_f(x-w)\end{pmatrix}.\]
Note that $u$ is constructed so that it spatially looks like the sum of a vertical and horizontal shear, each of $C^{\alpha-}$ regularity. We then consider the version of the Kraichnan flow to have the spatial covariance given by $D_s$ and to be white-in-time.

Then note that, suppressing the dependence on $\alpha$ and defining $g_s^\kappa$ exactly as $g^{\kappa,\eta,\alpha}$ was defined in Section~\ref{s.kraichnan-spec} but using the Kraichnan model with the drift given by this spatial covariance, we get that
\[\partial_t g_s^\kappa - \nabla \cdot a_s^\kappa \nabla g_s^\kappa = 0,\]
with
\[a_s^\kappa(x,y) := D_s(0,0) - D_s(x,y) +2\kappa I.\]
Note then that, similar to the bounds in Appendix~\ref{a.diffusion-matrix-computation}, we can compute that for $x \in [-\pi,\pi],$ 
\begin{equation}
    \label{eq.f-diff-ineq}
   D_f(0) - D_f(x) \geq c_\alpha |x|^{2\alpha}. 
\end{equation}
As such, uniformly in $\kappa$, we have that
\[\nabla g^\kappa \cdot a_s^\kappa \nabla g^\kappa \geq c_\alpha \left( |x|^{2\alpha} |\partial_y g^\kappa|^2 + |y|^{2\alpha} |\partial_x g^\kappa|^2\right).\]
Thus we can rerun the proof of anomalous diffusion in the Kraichnan model with this spatial covariance, using the exact same proof and getting the exact same results, provided we can replace the weighted Sobolev inequalities given in Propositions~\ref{prop.weighted-poincare-base} and~\ref{prop.est-for-nash} with the appropriate weighted Sobolev inequalities with nonradial weights. The necessary inequalities are provided by the following.

\begin{proposition}
    \label{prop.weighted-poincare-nonradial}
    Let $\gamma \in [0,1)$, $D:= [-\pi,\pi]^2,$ and $g : D \to \R$ such that $\int g(x)\ dx = 0$. Then there exists $C(\gamma) < \infty$ such that
    \[\|g\|_{L^2(D)} \leq C\left(\||x|^\gamma \partial_y g\|_{L^2(D)} + \||y|^\gamma \partial_x g\|_{L^2(D)}\right).\]
\end{proposition}

\begin{proposition}
    \label{prop.est-for-nash-nonradial}
    Let $\gamma \in [0,1)$, $D:= [-\pi,\pi]^2,$ and $g : D \to \R$ such that $\int g(x)\ dx = 0$. Then there exists $C(\gamma) < \infty$ such that 
    \[\|g\|_{L^2(D)} \leq C \left(\||x|^\gamma \partial_y g\|_{L^2(D)} + \||y|^\gamma \partial_x g\|_{L^2(D)}\right)^{1-a} \|g\|_{L^1(D)}^{a}.\]
    with
    \[a = \frac{1-\gamma}{2} \in (0,1).\]
\end{proposition}

Proofs of these inequalities are also provided in Appenidx~\ref{a.ineq-proofs}. Similar to Propositions~\ref{prop.weighted-poincare-base} and~\ref{prop.est-for-nash}, proofs are provided for completeness as no reference could be found, though it is possible inequalities of this sort are known to experts. A version of Proposition~\ref{prop.weighted-poincare-nonradial} stated on $\R^2$ for compactly supported function instead of zero-mean functions is given in the post~\cite{fedjahttps://mathoverflow.net/users/1131/fedjaAnswerCaffarelliKohnNirenbergtypeInequality2023}.

With these propositions in hand, the proof of the anomalous diffusion in this version of the Kraichnan model goes through verbatim as in Section~\ref{s.ad-kraichnan}, with~\eqref{eq.f-diff-ineq} in place of Proposition~\ref{prop.diffusion-matrix-bound}, Proposition~\ref{prop.weighted-poincare-nonradial} in place of Proposition~\ref{prop.weighted-poincare-base}, and Proposition~\ref{prop.est-for-nash-nonradial} in place of Proposition~\ref{prop.est-for-nash}.

Note additionally that in all of the proofs in Section~\ref{s.ad-kraichnan}, the only thing that was needed was control on $D(0) - D(x)$, and as such adding a constant matrix to $D$, which is equivalent to adding a constant drift times a temporal white-noise to $u$, doesn't affect the result.

From this discussion, we get the following proposition.
\begin{proposition}
    Fix $\alpha \in (0,1).$ Let $D_f : \T \to \R$ be defined by
    \begin{equation}    
    \label{eq.Df-def}
    \hat D_f(k) = \frac{1_{k \ne 0}}{|k|^{1+2\alpha}}
    \end{equation}
    and let $f,g : \R \times \T$ be independent centered Gaussian fields with covariance
    \[\E f(t,x) f(s,x') = \E g(t,x) g(s,x') = \delta(t-s) D_f(x-x').\]
    Let $u : \R \times \T^2$ be the centered Gaussian field given by the sum of the two shears that $f,g$ generate together with a constant drift times a white noise $dB_t$, so that
    \[u(t,x) = f(t,x) e_y + g(t,y) e_x + vdB_t,\]
    for some $v \in \R^d.$ Then $u$ has covariance given by 
    \begin{align}
    \notag
    \E u(t,x,y) \otimes u(t',x',y') &= \delta(t-t') \left(\begin{pmatrix} D_f(y-y') &0\\ 0 & D_f(x-x') \end{pmatrix} + v \otimes v\right)
    \\&=:\delta(t-s)\left(D_s(x-x', y-y') + v \otimes v\right).    \label{eq.Ds-def}
    \end{align}
    The Kraichnan SPDE associated to $u$ exhibits anomalous diffusion. In particular, there exists $C(\alpha,d)$ such that for all $\kappa>0$ and for any $\theta_0 \in L^2(\T^d)$ such that $\int \theta_0(x)\ dx =0,$
    if $\theta^{\kappa} : \T^d \to \R$ is the random function solving the Kraichnan SPDE 
    \[
    \begin{cases}
    d\theta^{\kappa}_t = \kappa \Delta \theta^{\kappa}_t - u \odot \nabla \theta^{\kappa}_t\\
    \theta^\kappa_0 = \theta_0,    
    \end{cases}
    \]
    then we have the estimate
    \begin{equation*} \E \|\theta^{\kappa}_t\|_{L^2(\T^d)}^2 \leq C e^{-t/C} \|\theta_0\|_{L^2(\T^d)}^2. \end{equation*}
\end{proposition}

\subsection{Second example: randomly oriented shears}
\label{ss.random-shears}

We are now ready to provide our second example of a correlated-in-time model which fails to exhibit anomalous diffusion despite its white-in-time limit being anomalously diffusive. In this model, we again take the drift field to be piecewise constant in time and iid on each distinct time interval. Roughly, the distribution for each time interval is given by 1) choosing a random $C^{\alpha-}$ Gaussian shear flow 2) independently randomly orienting it to be horizontal or vertical.

We will show that this distribution will have the model described in Subsection~\ref{ss.shear-kraichnan} as it's formal white-in-time limit (in the sense described in Subsection~\ref{ss.white-in-time-limits}). But for any fixed positive time correlation, any finite time interval will just consist of finitely many continuous shears, which will have unique ODE trajectories, and so will not exhibit anomalous diffusion by Theorem~\ref{thm.ad-imp-non-unique-ode}.

To be more precise, let $D_f$ be defined by~\eqref{eq.Df-def} and let $f: \T^d \to \R$ by the centered Gaussian field with covariance given by
\[\E f(x) f(y) = D_f(x-y).\]
Then let $f_j, j \in \N$ be iid copies of $f$. Let $B_j$ be iid Bernoulli random variables, that is $\P(B_j = 1) = \P(B_j = 0) = 1/2$. Then we define
\[u_j(x,y) := \sqrt{2}B_jf_j(x) e_y + \sqrt{2}(1-B_j) f_j(y) e_x,\]
and
\[u^\epsilon(t,x,y) := \epsilon^{-1/2} u_{\lceil t/\epsilon \rceil}(x,y).\]
Then we note that
\[\E u_j(x,y) \otimes u_j(x',y') = D_s(x-x', y-y'),\]
with $D_s$ as given by~\eqref{eq.Ds-def}. Note then that this spatial covariance is then the same as the model given in Subsection~\ref{ss.shear-kraichnan}, and as such the white-in-time limit of this model is anomalously diffusive.

On the other hand, it is a simple exercise to show ODE trajectories are unique in this model for each positive $\epsilon$ (note that the shears on each time interval are continuous). As such, by Theorem~\ref{thm.ad-imp-non-unique-ode}, for each positive $\epsilon$, there is no anomalous diffusion.

We have then shown the following proposition.

\begin{theorem}
    Let $u^\epsilon$ be the random field defined as above. Then, for each $\epsilon>0$ and for each realization, $u^\epsilon$ has unique ODE trajectories. As such, no realization of $u^\epsilon$ for positive $\epsilon$ exhibits anomalous diffusion. On the other hand, $u^\epsilon$ has as its (formal) white-in-time limit the drift field given in Subsection~\ref{ss.shear-kraichnan}, which does exhibit anomalous diffusion.
\end{theorem}

\subsection{Third example: sum of shears with a mean drift}

In the previous two examples, we were able to construct correlated-in-time models that did not exhibit anomalous diffusion by splitting the spatial distribution of the white-in-time model into distinct pieces and then having the correlated-in-time model only have one of those pieces active on each time interval. In the first model, we split into countably many smooth pieces. As such, for each realization at finite $\epsilon$, there were only finitely many ``scales'' interacting. In the second model, we split into only two H\"older continuous pieces, but by splitting into shears, we had for each realization and each finite $\epsilon$, there was only finitely many interactions between the horizontal shearing and the vertical shearing. In both these examples, we get anomalous diffusion in the white-in-time limit and this can heuristically seen as a consequence of their being infinitely many interactions between the different scales or different shear components in this limit, while for any finite amount of time correlation there are only finitely many such interactions.

In this final example, we don't split the field into distinct pieces causing there to be only finite interactions between distinct shear components. Yet we still can show that there fails to be anomalous diffusion at positive time correlation. As such, this example in some sense gives a more complex mechanism for the failure of the correlated-in-time model to be anomalously diffusive. Thus it provides fairly robust evidence that the presence or absence of anomalous diffusion in temporally correlated fields is a much more delicate property than it is for temporally uncorrelated fields.

In this example, we take the spatial structure of the field to be the sum of two shears together with a mean drift. We will show uniqueness of ODE trajectories for this field and as such will not get anomalous diffusion for the correlated-in-time model, but the white-in-time model will be the same as that of~\ref{ss.shear-kraichnan}

Let us now construct the spatial field we will be using. Note that, using a sine and cosine series, we can write the random $C^\alpha$ function $f: \T^d \to \R$ from Subsection~\ref{ss.random-shears} as 
\[f(x) = \sum_j c_j\phi_j(x) Z_j,\]
where the $Z_j$ are iid standard normal random variables, the $c_j$ are sequence of coefficients going to $0$ and the $\phi_j$ are the sines and cosines at the integer modes. In this example, we will need the random shear to be uniformly bounded, so we cannot build our example on Gaussians. We instead replace the standard normals $Z_j$ with $-1,1$ valued Bernoulli random variables $B_j$, so that $\E B_j = 0, \E B_j^2 =1$.\footnote{It is not hard to verify the below construction can be done with any bounded random variable provided it's centered and has unit variance.} So we let $g: \T \to\R$ be the random field given by
\[g(x) := \sum_j c_j \phi_j(x) B_j.\]
Then it's direct to verify that $g$ has the same covariance as $f$. We also have that for $\alpha>1/2,$ $g$ is uniformly bounded, as the $c_j$ are summable and so 
\[|g(x)| \leq \sum_j |c_j| =: K < \infty.\]

Then let $g_j, h_j, j \in \N$ be iid copies of $g$ and let $X_j$ be a sequence of iid $-1,1$ valued Bernoulli random variables. Then we let $u_j: \T^2 \to \R$ be defined by
\[u_j(x,y) = g_j(x) e_y + h_j(y) e_x + 2KX_j(e_x + e_y).\]
Then we note that 
\[\E u_j(x,y) \otimes u_j(x',y') = D_s(x-x', y-y') + 4K^2\begin{pmatrix} 1 & 1 \\ 1 &1\end{pmatrix}.\]
We let
\[u^\epsilon(t,x,y) := \epsilon^{-1/2} u_{\lceil t/\epsilon \rceil}(x,y).\]
Then the $\epsilon \to 0$ limit formally gives the Kraichnan model described in Subsection~\ref{ss.shear-kraichnan} with $v = 2K(e_x + e_y)$. As such, the white-in-time limit is anomalously diffusive.

So now we have to argue that for each positive $\epsilon$ and each realization, that $u^\epsilon$ fails to exhibit anomalous diffusion. Note that each $u_j$ is the sum of two shears and is constructed with a mean drift such that neither coordinate ever vanishes. Thus for any fixed $\epsilon$, $u^\epsilon$ is piecewise constant in time on finitely many intervals and on each interval nowhere vanishes. Thus from Corollary~\ref{cor.nonvanishing} we get that for positive $\epsilon$, $u^\epsilon$ does not exhibit anomalous diffusion.

\begin{theorem}
    Let $u^\epsilon$ be the random field defined as above. Then, no realization of $u^\epsilon$ for positive $\epsilon$ exhibits anomalous diffusion. On the other hand, $u^\epsilon$ has as its (formal) white-in-time limit the drift field given in Subsection~\ref{ss.shear-kraichnan}, which does exhibit anomalous diffusion.
\end{theorem}

\begin{acks}
    I would like to thank the following for stimulating discussions: Scott Armstrong and Vlad Vicol, on deterministic aspects of anomalous diffusion as well as writing advice; Theodore Drivas, for very helpful comments on a draft; Aria Halavati on weighted Sobolev inequalities, in particular on the proof of Proposition~\ref{prop.weighted-poincare-base}; Alexander Dunlap, on SPDE theory; Elias Hess-Childs, on stochastics and functional inequalities. The author was partially supported by NSF grants DMS-1954357 and DMS-2000200 as well as a Simons Foundation grant.
\end{acks}

\appendix

\section{Control on the diffusion matrix}
\label{a.diffusion-matrix-computation}

\begin{proof}[Proof of Proposition~\ref{prop.diffusion-matrix-bound}]
    Fix $\alpha \in (0,1)$ and without loss of generality take $|w| =1$. The $\rho$ cutoff makes it so that there are two different regimes that need to be treated separately, $|x| \leq \eta$ and $|x| \geq \eta$. We then need to combine the bounds on both regions. Thus the proof proceeds in three parts. First, in two parts, we prove the following bound
    \[w \cdot (D^\eta(0) - D^\eta(x)) w \geq \begin{cases} c \eta^{2\alpha -2} |x|^2 & |x| \leq \eta \\ c |x|^{2\alpha} & |x| \geq \eta,  \end{cases}\]
    where the constant is independent of $\alpha,\eta,w$. Lastly, we show that this bound implies the claimed bound.

    Before splitting into cases, note the following exact computation
    \begin{align*}
    w \cdot (D^\eta(0) - D^\eta(x)) w &=  \sum_{k \ne 0} (1- e^{ik\cdot x})\left(1 - \frac{(w\cdot k)^2}{|k|^2} \right) |k|^{-(d+2\alpha)} \rho(\eta |k|)
    \\&=   \sum_{k \ne 0} (1- \cos(k\cdot x))\left(1 - \frac{(w\cdot k)^2}{|k|^2} \right) |k|^{-(d+2\alpha)}\rho(\eta |k|)
    \\&=   {|x|^{2\alpha}} \sum_{k \ne 0} |x|^d(1- \cos(k \cdot x))\left(1 - (w\cdot \hat k)^2 \right) (|x||k|)^{-(d+2\alpha)}  \rho(\eta |k|) 
    \\&=  {|x|^{2\alpha}} \sum_{\substack{\zeta \ne 0,\\ \zeta \in |x|\Z^d}} |x|^d (1- \cos(\zeta \cdot \hat x))\left(1 - (w\cdot \hat \zeta)^2 \right) |\zeta|^{-(d+2\alpha)} \rho(\eta |x|^{-1}|\zeta| ).
    \end{align*}

    \textit{Step 1.} In this step, we consider the case the $|x| \leq \eta$. First note that following computation
    \begin{align*}
    w \cdot (D^\eta(0) - D^\eta(x)) w &=|x|^{2\alpha}\sum_{\substack{\zeta \ne 0,\\ \zeta \in |x|\Z^d}} |x|^d (1- \cos(\zeta \cdot \hat x))\left(1 - (w\cdot \hat \zeta)^2 \right) |\zeta|^{-(d+2\alpha)} \rho(\eta |x|^{-1}|\zeta| )
    \\&\geq c |x|^{2\alpha}\sum_{\substack{\zeta \ne 0,\\ \zeta \in |x|\Z^d,\\ |\zeta| \leq 1}} |x|^d |\zeta \cdot \hat x|^2\left(1 - (w\cdot \hat \zeta)^2 \right) |\zeta|^{-(d+2\alpha)} \rho(\eta |x|^{-1}|\zeta|)
    \\&\geq c |x|^{2\alpha}\sum_{\substack{\zeta \ne 0,\\ \zeta \in |x|\Z^d,\\ |\zeta| \leq 1}} |x|^d 1_{\{|\hat \zeta \cdot \hat x| \geq \cos(\pi/3), |w \cdot \hat \zeta| \leq \cos(\pi/6) \}}  |\zeta|^{-d +2-2\alpha} \rho(\eta |x|^{-1}|\zeta|).
    \end{align*}
    We note then that the indicator amounts to restriction on the angles that $\zeta$ can occupy and the remaining terms are purely radially. We also have that in $d \geq 2$
    \[\inf_{\hat x, w \in S^{d-1}} |\{z \in S^{d-1} : |z \cdot \hat x| \geq \cos(\pi/3), |z \cdot w| \leq \cos(\pi/6)\}| >0,\]
    as such we are restricting over a set of angles of uniformly lower bounded measure. If instead of a sum over a lattice, we were integrating over $\R^d$, we could then remove the angular restriction at the cost of an additional constant. The lattice makes things a bit trickier, but by additionally restricting the sum to $|\zeta| \geq a |x|/\eta$, with $a \in (0,1)$ to be later specified uniformly in $\eta$, then for small enough $\eta$, we can also remove the angular restriction at the price of a uniform constant,  
    \begin{align*}
    \sum_{\substack{\zeta \ne 0,\\ \zeta \in |x|\Z^d,\\ |\zeta| \leq 1}}& |x|^d 1_{\{|\hat \zeta \cdot \hat x| \geq \cos(\pi/3), |w \cdot \hat \zeta| \leq \cos(\pi/6) \}}  |\zeta|^{-d +2-2\alpha} \rho(\eta |x|^{-1}|\zeta|)
    \\&\geq c \sum_{\substack{\zeta \in |x|\Z^d,\\ a |x|/\eta \leq |\zeta| \leq 1}} |x|^d   |\zeta|^{-d +2-2\alpha} \rho(\eta |x|^{-1}|\zeta|).
    \end{align*}
    We can do this because as $\eta$ gets small, the lattice points $|x| \Z^d$ more densely the fill the angles in a thin spherical shells with radii greater than $a|x|/\eta$. Making this precise is elementary but tedious. 

    As such, we have that for $|x| \leq \eta$ and $\eta$ sufficiently small, 
    \[w \cdot (D^\eta(0) - D^\eta(x)) w \geq c|x|^{2\alpha}\sum_{\substack{\zeta \in |x|\Z^d,\\ a |x|/\eta \leq |\zeta| \leq 1}} |x|^d |\zeta|^{-d +2-2\alpha} \rho(\eta |x|^{-1}|\zeta|).\]
    Then, as the summand is radially decaying, we can replace the sum with an integral (noting the $|x|^d$ is a volume factor for the lattice cells) to get a lower bound, giving
    \begin{align*}
        w \cdot (D^\eta(0) - D^\eta(x)) w &\geq c |x|^{2\alpha} \int_{a |x| /\eta \leq |y| \leq 1} |y|^{-d+2-2\alpha} \rho(\eta |x|^{-1} |y|)
        \\&= c |x|^{2\alpha} \int_{a|x|/\eta}^1 r^{1-2\alpha} \rho(\eta |x|^{-1} r).
    \end{align*}
    Now recall that $\rho(0) = 1$, $\rho$ is smooth and radially decaying. So let $b \in (0,1)$, depending on $\rho$, such that $\rho(b) \geq 1/2$ and let $a = b/2$. Then we get that
    \begin{align*}
        w \cdot (D^\eta(0) - D^\eta(x)) w &\geq c |x|^{2\alpha} \int_{b|x|/2\eta}^{b|x|/\eta} r^{1-2\alpha} = c\eta^{2\alpha -2}|x|^{2},
    \end{align*}
    with $c$ depending on $\rho$ but independent of $\eta,\alpha$. Thus we conclude Step~1.

    \textit{Step 2.} We now consider the case that $|x| \geq \eta$. In which case, using that $\rho$ is radially decaying, we have that 
    \begin{align*}
        w \cdot (D^\eta(0) - D^\eta(x)) w &=  |x|^{2\alpha}\sum_{\substack{\zeta \ne 0,\\ \zeta \in |x|\Z^d}} |x|^d (1- \cos(\zeta \cdot \hat x))\left(1 - (w\cdot \hat \zeta)^2 \right) |\zeta|^{-(d+2\alpha)} \rho(\eta |x|^{-1}|\zeta| )
        \\&\geq c |x|^{2\alpha}\sum_{\substack{\zeta \ne 0,\\ \zeta \in |x|\Z^d}} |x|^d (1- \cos(\zeta \cdot \hat x))\left(1 - (w\cdot \hat \zeta)^2 \right) |\zeta|^{-(d+2\alpha)} \rho(|\zeta| ).
    \end{align*}
    Then let $b>0$ such that $\rho(b) =1/2$. Then 
    \begin{align*}
        w \cdot (D^\eta(0) - D^\eta(x)) w &\geq c |x|^{2\alpha}\sum_{\substack{b/2 \leq |\zeta| \leq b,\\ \zeta \in |x|\Z^d}} |x|^d (1- \cos(\zeta \cdot \hat x))\left(1 - (w\cdot \hat \zeta)^2 \right) |\zeta|^{-(d+2\alpha)} 
        \\&\geq  c |x|^{2\alpha}\sum_{\substack{b/2 \leq |\zeta| \leq b,\\ \zeta \in |x|\Z^d}} |x|^d |\zeta|^{-(d+2\alpha)} 
        \\&\geq c |x|^{2\alpha} \int_{b/2}^b r^{-1-2\alpha}\ dr
        \\&\geq c|x|^{2\alpha}, 
    \end{align*}
    where the second inequality that removes the dependence on the angle of $\zeta$ follows exactly as in Step 1, and the switch from sum to integral also follows as in Step 1. The constant depends on $\rho$ but not $\alpha,\eta$. Thus we conclude Step 2.

    \textit{Step 3.} The above steps imply that 
    \[w \cdot a^{\kappa,\eta,\alpha}(x) w \geq \begin{cases} c\eta^{2\alpha -2} |x|^2 + \kappa & |x| \leq \eta \\ c |x|^{2\alpha} + \kappa & |x| \geq \eta.\end{cases}\]
    We seek to uniformly lower bound this by a small multiple of $|x|^{2\beta}$ for some $\beta \in [\alpha,1]$. The bound then is straightforward in the regime that $|x| \geq \eta$. For $|x| \leq \eta$, we compute the minimum of 
    \[\eta^{2\alpha -2} |x|^{2-2\beta} + \kappa |x|^{-2\beta},\]
    which gives that 
    \[\eta^{2\alpha -2} |x|^2 + \kappa \geq  \Big(\Big(\frac{\beta}{1-\beta}\Big)^{1-\beta} + \Big(\frac{1-\beta}{\beta}\Big)^{\beta}\Big)\eta^{2\beta(\alpha-1)} \kappa^{1-\beta} |x|^{2\beta} \geq \eta^{2\beta(\alpha-1)} \kappa^{1-\beta} |x|^{2\beta},\]
    and hence, for any $\beta \in [\alpha,1)$, 
    \[w \cdot a^{\kappa,\eta,\alpha}(x) w \geq c (1 \land \eta^{2\beta (\alpha-1)} \kappa^{1-\beta}) |x|^{2\beta},\]
    allowing us to conclude.
\end{proof}

\section{Proofs of weighted inequalities}
\label{a.ineq-proofs}

We note that in $d=2$, Propositions~\ref{prop.weighted-poincare-base} and~\ref{prop.est-for-nash} are corollaries of Propositions~\ref{prop.weighted-poincare-nonradial} and~\ref{prop.est-for-nash-nonradial}. We provide separate proofs of them here for two reasons. First, it is simpler to provide complete proofs here then it is to adapt Propositions~\ref{prop.weighted-poincare-nonradial} and~\ref{prop.est-for-nash-nonradial} to arbitrary dimension, which would be needed to get the wanted arbitrary dimension in Propositions~\ref{prop.weighted-poincare-base} and~\ref{prop.est-for-nash}. Second, only Propositions~\ref{prop.weighted-poincare-base} and~\ref{prop.est-for-nash} are needed for proof of anomalous diffusion in the usual Kraichnan model and in an effort to give a self contained and minimal proof of that fact, we wish to avoid relying on Propositions~\ref{prop.weighted-poincare-nonradial} and~\ref{prop.est-for-nash-nonradial}.

\begin{proof}[Proof of Proposition~\ref{prop.weighted-poincare-base}]
    Before proceeding, let us recall the version of the Caffarelli-Kohn-Nirenberg inequality we are proving an alternative version of, which states that for $g \in C_c^\infty(\R^d)$,
\begin{equation}\|g\|_{L^2} \leq C \||x| \nabla g\|_{L^2}.
\label{eq.CKN-inequality-prop-weighted-poincare}\end{equation}
    We will be using this version of the inequality in this proof.

    We prove the proposition by contradiction, similarly to the usual proof of Poincar\'e. Assuming the inequality fails, by normalizing we can construct a sequence $g_n$ such that
    \[\int g_n= 0;\qquad \|g_n\|_{L^2} = 1;\qquad \||x| \nabla g_n\|_{L^2} \to 0.\]
    By the $L^2$ boundedness, we have that the sequence $g_n$ is weakly compact in $L^2$. Thus be reindexing we can take without loss of generality that $g_n \stackrel{L^2}{\rightharpoonup} g.$

    We first claim that $g = 0$. First note that using weak convergence against $1$,
    \[0 = \int 1 g_n \to  \int 1 g = \int g.\]
    We next claim that the distributional derivative $\nabla g$ is distributionally equal to $0$ on $[-\pi,\pi]^d \backslash \{0\}$. Before proving this claim note this suffices to see that $g =0$. Since if the distributional derivative on $[-\pi,\pi]^d \backslash \{0\}$ is equal to $0$, by the connectedness of this set, $g$ is a constant on this set, and then by $g \in L^2$, $g$ is a constant on $[-\pi,\pi]^d$. Then since $\int g =0$, that constant must be $0$.

    Let's now compute the distributional derivative. Take $\phi \in C_c^\infty([-\pi,\pi]^d \backslash \{0\})$. Then let
    \[\psi := \frac{\phi}{|x|} \in  C_c^\infty((-\pi,\pi)^d \backslash \{0\}).\]
    Then
    \[\abs{\int \phi \nabla g} \leq \lim_n \abs{\int \phi \nabla g_n} = \lim_n \abs{\int |x| \psi \nabla g_n} \leq \limsup_n \|\psi\|_{L^2} \||x| \nabla g_n\|_{L^2} =0.\]
    Thus $\nabla g = 0$ distributionally, so $g=0$ by the above argument.
    
    We are now prepared to show the contradiction. Let $\chi : D \to \R$ be the $W^{1,\infty}$ piecewise-affine cutoff between $[-\pi/2,\pi/2]^d$ and $[-\pi,\pi]^d$. Then note that 
    \[1 = \|g_n\|_{L^2}^2 = \|\chi g_n\|_{L^2}^2 +\int (1- \chi^2) g_n^2.\]
    Then note that $\chi g_n$ is compactly supported on $\R^d$, as such we can apply the inequality~\eqref{eq.CKN-inequality-prop-weighted-poincare}, giving 
    \begin{align*}
    \|\chi g_n\|_{L^2}^2 &\leq C \||x|\nabla (\chi g_n)\|_{L^2}^2 
    \\&\leq C \||x| \chi \nabla g_n\|_{L^2}^2 + C \||x| \nabla \chi g_n\|_{L^2}^2
    \\&\leq C \||x| \nabla g_n\|_{L^2}^2 + C \|\nabla \chi g_n\|_{L^2}^2.
    \end{align*}
    Putting the two displays together, we get
    \begin{equation}\label{eq.contra-CKN-poinc}1 \leq C\||x| \nabla g_n\|_{L^2}^2 + C \int (1-\chi^2 + |\nabla \chi|^2) g_n^2 \leq C\||x| \nabla g_n\|_{L^2}^2 + C \int_{[-\pi,\pi]^d \backslash [-\pi/2,\pi/2]^d}g_n^2.
    \end{equation}
    Note that the first term on the right hand side goes to $0$ by assumption. We now want to argue the second term (subsequentially) also goes to $0$, giving the desired contradiction.

    Let $\tilde D := [-\pi,\pi]^d \backslash [-\pi/2,\pi/2]^d$. Note that
    \[\|g_n\|_{L^2(\tilde D)} \leq \|g_n\|_{L^2(D)} = 1\]
    and
    \[\|\nabla g_n\|_{L^2(\tilde D)} \leq C \||x| \nabla g_n\|_{L^2(\tilde D)} \leq C \||x| \nabla g_n\|_{L^2(D)} \to 0.\]
    Thus $g_n$ is a sequence uniformly bounded in $H^1(\tilde D).$ By Rellich-Kondrachov, $g_n$ is compact in $L^2(\tilde D)$, thus (after reindexing) $g_n \to v$ in $L^2(\tilde D)$. Recall $g_n \rightharpoonup 0$ in $L^2(D)$ and so also $g_n \rightharpoonup 0$ in $L^2(\tilde D)$. Then since weak and strong limits agree, $v =0,$ so $g_n \to 0$ in $L^2(\tilde D)$. We then conclude the argument, as we now see both terms on the right hand side of~\eqref{eq.contra-CKN-poinc} are going to $0$, contradicting that the sum is greater than $1$.
\end{proof}

\begin{proof}[Proof of Proposition~\ref{prop.est-for-nash}]
    We recall the Caffarelli-Kohn-Nirenberg inequality we will be using, which states that for $v \in C_c^\infty(\R^d)$, $\beta\in (0,1)$, 
    \begin{equation}
    \|v\|_{L^2(\R^d)} \leq C \| |x|^\beta \nabla v\|_{L^2(\R^d)}^a \|v\|_{L^1(\R^d)}^{1-a}
    \label{eq.CKN-inequality-prop-est-for-nash}    
    \end{equation}
    with
    \[a = \frac{d}{d+2-2\beta}.\]

    Identify $g$ with a periodic function $\R^d \to \R$ and the torus with $[-\pi,\pi]^d$. Let $\chi$ be the piecewise affine cutoff between $[-\pi,\pi]^d$ and $[-3\pi/2,3\pi/2]^d$. Then $\chi g$ is compactly supported on $\R^d$ so we can apply the inequality above~\eqref{eq.CKN-inequality-prop-est-for-nash}, giving
    \begin{align*}
        \|g\|_{L^2(\T^d)} &\leq \|\chi g\|_{L^2(\R^d)} 
        \\&\leq C \||x|^\beta \nabla (\chi g)\|_{L^2(\R^d)}^a \|\chi g\|_{L^1(\R^d)}^{1-a} 
        \\&\leq C \left( \|\chi |x|^\beta \nabla g\|_{L^2(\R^d)}^a + \||x|^\beta g\nabla \chi \|_{L^2(\R^d)}^a \right)\|g\|_{L^1(\T^d)}^{1-a}
        \\&\leq C \left( \||x|^\beta \nabla g\|_{L^2(\T^d)}^a + \|g\|_{L^2(\T^d)}^a \right)\|g\|_{L^1(\T^d)}^{1-a},
    \end{align*}
    giving the desired bound.
\end{proof}

The proofs of Propositions~\ref{prop.weighted-poincare-nonradial} and~\ref{prop.est-for-nash-nonradial} are more extensive. First, some notation.

\begin{definition}
    Let $\alpha,\beta \in \R$ and $p,q \in [1,\infty)$, then for $S \subseteq \R$, $g : S \to \R$, let
    \begin{align*}\|g\|_{L^{p,\alpha}(S)} := \left(\int_S |x|^\alpha |g(x)|^p\ dx\right)^{1/p}.\end{align*}
\end{definition}

\begin{lemma}
    For $\alpha \in [0,1),$
    \begin{align*}1 \leq p < \frac{2}{\alpha + 1},\end{align*}
    and $r>0$, there exists $C(r,\alpha,p) < \infty$ such that for any $g : \R^2 \supseteq B_r \to \R$,
    \begin{align*}\|g\|_{L^p} \leq C\|g\|_{L^\infty_x L^{1,\alpha}_y}^{1/2} \|g\|_{L^{1,\alpha}_x L^\infty_y}^{1/2},\end{align*}
    where all integrals are taken over $B_r$.
\end{lemma}

\begin{proof}
    Note that
    \begin{align*}\|g\|_{L^p}  &= \|\|g\|_{L^p_y} \|_{L^p_x} 
    \\& = \|\||y|^{-\alpha/2} |y|^{ \alpha/2} |g|^{1/2}|g|^{1/2} \|_{L^p_y}\|_{L^p_x}
    \\&\leq \||y|^{-\alpha/2}\|_{L_y^{2p/(2-p)}([-r,r])} \| \|g\|_{L^\infty_y}^{1/2} \|(|y|^\alpha |g|)^{1/2}\|_{L^2_y} \|_{L^p_x}
    \\&\leq C_{\alpha,p,r} \| \|g\|_{L^\infty_y}^{1/2} \||y|^\alpha g\|_{L^1_y}^{1/2} \|_{L^p_x} 
    \\&\leq  C\| |x|^{-\alpha/2}|x|^{\alpha/2}\|g\|_{L^\infty_y}^{1/2} \||y|^\alpha g\|_{L^1_y}^{1/2} \|_{L^p_x}
    \\&\leq C \||x|^{-\alpha/2}\|_{L_x^{2p/(2-p)}([-r,r])} \|g\|_{L^{1,\alpha}_x L^\infty_y}^{1/2} \|g\|_{L^\infty_x L^{1,\alpha}_y}^{1/2}
    \\&\leq C\|g\|_{L^{1,\alpha}_x L^\infty_y}^{1/2} \|g\|_{L^\infty_x L^{1,\alpha}_y}^{1/2},\end{align*}
    where we use the restriction on $p$ to get that 
    \begin{align*}\||y|^{-\alpha/2}\|_{L_y^{2p/(2-p)}([-r,r])} <\infty.\end{align*}
\end{proof}

\begin{lemma}
    Let $g \in C_c^\infty((-\infty,0]^2)$, then
    \begin{align*}\|g\|_{L_x^{1,\alpha} L_y^\infty} \leq \|\partial_y g\|_{L^{1,\alpha}_x L^1_y}\end{align*}
    and
    \begin{align*}\|g\|_{L_x^\infty L^{1,\alpha}_y} \leq \|\partial_x g\|_{L^1_x L^{1,\alpha}_y}.\end{align*}
\end{lemma}
\begin{note}
    The requirement $g \in C_c^\infty((-\infty,0]^2)$ means the support of $g$ can contains parts of the $x$ and $y$-axes, but $g$ must eventually be $0$ in the lower-left quadrant.
\end{note}

\begin{proof}
    For fixed $x,y \in \R$,
    \begin{align*}|g(x,y)| = \abs{ \int_{-\infty}^y \partial_y g(x,r)\ dr} \leq \|\partial_y g\|_{L^1_y}(x)\end{align*}
    Thus
    \begin{align*}\|g\|_{L^\infty_y}(x) \leq \|\partial_y g\|_{L^1_y}(x).\end{align*}
    Thus by monotonicity,
    \begin{align*}\|g\|_{L^{1,\alpha}_x L^\infty_y} \leq \|\partial_y g\|_{L^{1,\alpha}_x L^1_y}.\end{align*}

    For the second inequality, we note that for fixed $x$, we have that
    \begin{align*}\|g\|_{L^{1,\alpha}_y}(x) &= \int |y|^{\alpha} |g(x,y)|\ dy 
    \\&= \int_{-\infty}^x \frac{d}{dr} \left(\int |y|^{\alpha} |g(r,y)|\ dy\right) dr 
    \\&\leq \int_{-\infty}^x \int |y|^{\alpha} |\partial_x g|(r,y)\ dydr 
    \\&\leq \|\partial_x g\|_{L^1_x L^{1,\alpha}_y}.\end{align*}
    Thus taking the $L^\infty$ norm over $x$, we conclude.
\end{proof}

The direct consequence of the above lemmas is the following.

\begin{proposition}
    For $\alpha \in [0,1),$
    \begin{align*}1 \leq p < \frac{2}{\alpha + 1},\end{align*}
    and $r>0$, there exists $C(r,\alpha,p) < \infty$ such that for any $g \in C_c^\infty((-r,0]^2)$,
    \begin{align*}\|g\|_{L^p} \leq  C\|\partial_y g\|_{L^{1,\alpha}_x L^1_y}^{1/2}\|\partial_x g\|_{L^1_x L^{1,\alpha}_y}^{1/2}.\end{align*}
\end{proposition}

By applying this proposition to $v = |g|^\gamma$ for $\gamma>1$, we can get a similar inequality for a broader range of norms. In particular, we have the following proposition.

\begin{proposition}
    For $\alpha \in [0,1),$ 
    \begin{align*}1 \leq q < 2\alpha^{-1},\end{align*}
    and $r >0$, there exists $C(r,\alpha,q) < \infty$ such that for any $g \in C_c^\infty((-r,0]^2)$, 
    \begin{equation}
    \label{eq.zero-trace-ineq}
    \|g\|_{L^q} \leq C  \||x|^\alpha\partial_y g\|_{L^2}^{1/2} \||y|^\alpha \partial_x g\|_{L^2}^{1/2}.
    \end{equation}
\end{proposition}

\begin{proof}
    Fix $\gamma \geq 1$ and apply the above proposition to $|g|^\gamma$, yielding
    \begin{align*}\|g\|_{L^{\gamma p}}^\gamma &= \||g|^\gamma\|_{L^p} 
    \\&\leq C\gamma  \||g|^{\gamma -1} \partial_y g\|_{L^{1,\alpha}_x L^1_y}^{1/2} \||g|^{\gamma-1} \partial_x g\|_{L^1_x L^{1,\alpha}_y}^{1/2}
    \\&\leq C\gamma \||g|^{\gamma-1}\|_{L^2} \||x|^\alpha\partial_y g\|_{L^2}^{1/2} \||y|^\alpha \partial_x g\|_{L^2}^{1/2}
    \\&=  C\gamma \|g\|_{L^{2(\gamma-1)}}^{\gamma-1} \||x|^\alpha\partial_y g\|_{L^2}^{1/2} \||y|^\alpha \partial_x g\|_{L^2}^{1/2}.\end{align*}
    In order to apply the above proposition, we needed
    \begin{align*}1 \leq p < \frac{2}{\alpha+1}.\end{align*}
    We take $\gamma$ such that $\gamma p = 2(\gamma-1),$ hence
    \begin{align*}\gamma :=  \frac{2}{2 -p},\end{align*}
    in which case we get that
    \begin{align*}\|g\|_{L^{\frac{2p}{2-p}}} =\|g\|_{L^{\gamma p}} \leq C \gamma \||x|^\alpha\partial_y g\|_{L^2}^{1/2} \||y|^\alpha \partial_x g\|_{L^2}^{1/2}.\end{align*}
    Note $\gamma \geq 1$ as 
    \begin{align*}p < \frac{2}{\alpha+1} \leq 2.\end{align*}
    Then we are free to choose $p$ as we want in the range $[1,2/(\alpha+1))$ (getting a $p$ dependent constant) and noting that $2p/(2-p)$ is increasing in $p$ and that
    \begin{align*}\frac{2 \frac{2}{\alpha+1}}{2- \frac{2}{\alpha+1}} =2 \alpha^{-1},\end{align*}
    we get that for any $q \in [1,2\alpha^{-1})$,
    \begin{align*}\|g\|_{L^q} \leq C_{q,\alpha,r} \||x|^\alpha\partial_y g\|_{L^2}^{1/2} \||y|^\alpha \partial_x g\|_{L^2}^{1/2},\end{align*}
    giving the desired bound.
\end{proof}

We will want to prove this inequality for mean-zero functions instead of trace-zero functions. To that end, we first prove the following lemma.

\begin{lemma}
    \label{lem.zero-trace-alt}

    For $\alpha \in [0,1),$
    \begin{align*}1 \leq p < \frac{2}{\alpha + 1},\end{align*}
    then there exists $C(p,\alpha)<\infty$ such that for any $g \in C^\infty([-1,0]^2)$ such that $g(-1,y) =0 $ for all $y \in [-1,0]$, we have the estimate
    \begin{align*}\|g\|_{L^p} \leq C(\||x|^\alpha \partial_y g\|_{L^1} + \|\partial_x g\|_{L^1}).\end{align*}
\end{lemma}

\begin{proof}
    Note that for each $x \in [-1,0]$,
    \begin{align*}\int_{-1}^0 |g(x,y)|\ dy = \int_{-1}^x \frac{d}{ds} \int_{-1}^{0} |g(s,y)|\ dy\,ds \leq \|\partial_x g\|_{L^1}.\end{align*}
    Thus
    \begin{align*}\|g\|_{L^\infty_x L^1_y} \leq \|\partial_x g\|_{L^1}.\end{align*}
    Then for each $x \in [-1,0]$, there exists $y_x \in [-1,0]$ such that
    \begin{align*}|g(x,y_x)| \leq \|g\|_{L^1_y([-1,0])}(x) \leq \|g\|_{L^\infty_x L^1_y} \leq \|\partial_x g\|_{L^1}.\end{align*}
    Thus for each $x \in [-3/4,0]$, 
    \begin{align*}|g(x,y)| \leq \abs{\int_{y_x}^y \partial_y g(x,s)\ ds} + |g(x,y_x)| \leq \int_{-1}^{0} |\partial_y g(x,s)|\ ds +  \|\partial_x g\|_{L^1}.\end{align*}
    Thus
    \begin{align*}\|g\|_{L^\infty_y}(x) \leq \int_{-1}^0|\partial_y g(x,s)|\ ds + \|\partial_x g\|_{L^1}.\end{align*}
    Multiplying by $|x|^\alpha$ and integrating over $x$ gives
    \begin{align*}\|g\|_{L^{1,\alpha}_x L^\infty_y} \leq \|\partial_y g\|_{L^{1,\alpha}_x L^1_y}  + C \|\partial_x g\|_{L^1}.\end{align*}
    Then
    \begin{align*}\|g\|_{L^p}  &= \|\||g|^{1/2}|g|^{1/2}\|_{L^p_y} \|_{L^p_x} 
    \\&\leq \| \|g\|_{L^\infty_y}^{1/2} \|g\|_{L^1_y}^{1/2}\|_{L^p_x}
    \\&\leq \||x|^{-\alpha/2}\|_{L_x^{2p/(2-p)}([-1,0])} \|g\|_{L^{1,\alpha}_x L^\infty_y}^{1/2} \|g\|_{L^\infty_x L^1_y}^{1/2}
    \\&\leq C (\|g\|_{L^{1,\alpha}_x L^\infty_y} + \|g\|_{L^\infty_x L^1_y})
    \\&\leq C(\|\partial_y g\|_{L^{1,\alpha}_x L^1_y} + \|\partial_x g\|_{L^1}),
    \end{align*}
    thus giving the desired bound.
\end{proof}

\begin{proposition}
    For $\alpha \in [0,1),$ 
    \begin{align*}1 \leq q < 2\alpha^{-1},\end{align*}
    there exists $C(q,\alpha)$ such that for any $g \in C^\infty([-1,0]^2)$ such that
    \begin{align*}\int_{[-1,-1/2]^2} g = 0,\end{align*}
    then
    \begin{align*}\|g\|_{L^q} \leq C(\||x|^\alpha\partial_y g\|_{L^2}+ \||y|^\alpha \partial_x g\|_{L^2}).\end{align*}
\end{proposition}

\begin{proof}
    We assume without loss of generality that $q \geq 2$. Then note that it suffices to prove the result under the assumption that $g|_{[-1,-3/4]^2}=0$. To see this, suppose $\int_{[-1,-1/2]^2} g = 0$ and let $\chi$ be a smooth cutoff between $[-1,0]^2 \backslash [-1,-1/2]^2$ and $[-1,0]^2 \backslash [-1,-3/4]^2$, so that $v := \chi g$ is such that $v|_{[-1,-3/4]^2} =0$. As such by assumption we can apply the inequality to $v$. Thus we have that
    \begin{align*}\|g\|_{L^q} &\leq \|g\|_{L^q([-1,-1/2]^2)} + \|v\|_{L^q} 
    \\&\leq  \|g\|_{L^q([-1,-1/2]^2)} + C(\||x|^\alpha\partial_y (\chi g)\|_{L^2}+ \||y|^\alpha \partial_x (\chi g)\|_{L^2})
    \\&\leq  C(\|g\|_{L^q([-1,-1/2]^2)} + \||x|^\alpha\partial_y g\|_{L^2}+ \||y|^\alpha \partial_x g\|_{L^2}).\end{align*}
    But then we note that the usual Poincar\'e inequality together with the fact that $\int_{[-1,-1/2]^2} g = 0$ allows us to get the bound 
    \begin{align*}\|g\|_{L^q([-1,-1/2]^2)} \leq C \|\nabla g\|_{L^q([-1,-1/2])}^2 \leq C(\||x|^\alpha\partial_y g\|_{L^2}+ \||y|^\alpha \partial_x g\|_{L^2}).\end{align*}
    Putting this together yields 
    \begin{align*}\|g\|_{L^q} \leq C(\||x|^\alpha\partial_y g\|_{L^2}+ \||y|^\alpha \partial_x g\|_{L^2}),\end{align*}
    as desired.

    So we thus suppose that $g|_{[-1,-3/4]^2} =0$. We want to use the previously proven inequality~\eqref{eq.zero-trace-ineq}, so we cut off to make $g$ compactly supported in $(-1,0]^2$. To that end, let $\psi: [-1,0] \to \R$ be a smooth cutoff between $[-3/4,0]$ and $[-1,0]$ and let $v(x,y) := \psi(x)\psi(y) g(x,y).$ Then $v \in C_c^\infty((-1,0]^2)$, so we can apply the above inequality. Letting $A:= [-1,0]^2 \backslash [-3/4,0]^2$, we have
    \begin{align*}\|g\|_{L^q} &\leq \|g\|_{L^q(A)} + \|v\|_{L^q} 
    \\&\leq \|g\|_{L^q(A)} + C \||x|^\alpha\partial_y v\|_{L^2}^{1/2} \||y|^\alpha \partial_x v\|_{L^2}^{1/2}
    \\&\leq \|g\|_{L^q(A)} + C (\||x|^\alpha\partial_y v\|_{L^2}+ \||y|^\alpha \partial_x v\|_{L^2})
    \\&\leq C (\|g\|_{L^q(A)} + \||x|^\alpha\partial_y g\|_{L^2}+ \||y|^\alpha \partial_x g\|_{L^2}).\end{align*}
    Then to conclude, it suffices to show note the bound
    \begin{align*}\|g\|_{L^q(A)}  \leq C(\||x|^\alpha\partial_y g\|_{L^2}+ \||y|^\alpha \partial_x g\|_{L^2})\end{align*}
    follows from Proposition~\ref{lem.zero-trace-alt} using that $g|_{[-1,-3/4]^2} =0$.
\end{proof}

We are now ready to prove the weighted inequality for zero mean functions.

\begin{proposition}
    \label{prop.mean-zero-ineq-nonradial}
    For $\alpha \in [0,1),$ 
    \begin{align*}1 \leq q < 2\alpha^{-1},\end{align*}
    there exists $C(q,\alpha)$ such that for any $g \in C^\infty([-\pi,\pi]^2)$ such that
    \begin{align*}\int g = 0,\end{align*}
    then
    \begin{align*}\|g\|_{L^q} \leq C(\||x|^\alpha\partial_y g\|_{L^2}+ \||y|^\alpha \partial_x g\|_{L^2}).\end{align*}
\end{proposition}

\begin{proof}
    We again can assume without loss of generality that $q\geq2$. We prove this inequality by contradiction. Suppose that $g_n \in C^\infty([-1,1]^2)$ such that
    \begin{align*}\|g_n\|_{L^q} = 1;\qquad \int g_n = 0;\qquad \||x|^\alpha\partial_y g_n\|_{L^2}+ \||y|^\alpha \partial_x g_n\|_{L^2} \to 0.\end{align*}
    Then by weak compactness and relabelling the subsequence, we can assume
    \begin{align*}g_n \stackrel{L^q}{\weakto} g.\end{align*}
    Then in particular
    \begin{align*}0 =\int 1 g_n \to \int 1 g  = \int g.\end{align*}
    Further, one can easily verify that $\partial_x g_n \to 0$ in $L^2_{loc}([-1,1] \times [-1,0) \cup (0,1])$ and $\partial_y g_n \to 0$ in $L^2_{loc}([-1,0)\cup (0,1] \times [-1,1]).$

    Thus the distribution $\partial_x g = 0$ away from $y=0$ and the distribution $\partial_y g = 0$ away from $x=0$. This then implies that $g$ is a constant. But $\int g =0 $, so $g =0$. Thus
    \begin{align*}g_n \stackrel{L^q}{\weakto} 0.\end{align*}
    We will now show that
    \begin{align*}\|g_n\|_{L^q([-1,0]^2)} \to 0,\end{align*}
    the other quadrants follow similarly. Note that this suffices to finish the proof, as we then get that $g_n \to 0$ in $L^q$, contradicting, $\|g_n\|_{L^q} = 1$.

    Let
    \begin{align*}a_n := 4 \int_{[-1,-1/2]^2} g_n,\end{align*}
    and note by the weak convergence $a_n \to 0$. Let $v_n := g_n - a_n$. Then note that
    \begin{align*}\int_{[-1,-1/2]^2} v_n =0,\end{align*}
    so we can apply the previously proven inequality to give that
    \begin{align*}\|g_n\|_{L^q([-1,0]^2)} &\leq a_n + \|v_n\|_{L^q([-1,0]^2)} 
    \\&\leq a_n + C(\||x|^\alpha\partial_y g_n\|_{L^2([-1,0]^2)}+ \||y|^\alpha \partial_x g_n\|_{L^2([-1,0]^2)}) \to 0,\end{align*}
    allowing us to conclude.
\end{proof}

\begin{proof}[Proof of Proposition~\ref{prop.weighted-poincare-nonradial} and Proposition~\ref{prop.est-for-nash-nonradial}]
    Note now that Proposition~\ref{prop.weighted-poincare-nonradial} is just Proposition~\ref{prop.mean-zero-ineq-nonradial} applied with $q=2$.

    For Proposition~\ref{prop.est-for-nash-nonradial}, let
    \begin{align*}q = \frac{2\alpha^{-1} + 2}{2} = \frac{\alpha + 1}{\alpha}.\end{align*}
    Then $2 < q < 2\alpha^{-1}$. We then interpolate $2$ between $q$ and $1$, giving
    \begin{align*}\|g\|_{L^2} \leq \|g\|_{L^q}^{1-a} \|g\|_{L^1}^a\end{align*}
    with
    \begin{align*}a + \frac{1-a}{q} = \frac{1}{2}\end{align*}
    or
    \begin{align*}a= \frac{q-2}{2q-2} = \frac{1-\alpha}{2}.\end{align*}
    Then applying Proposition~\ref{prop.mean-zero-ineq-nonradial} to bound $\|g\|_{L^q}$, we conclude.
\end{proof}

\section{Properties of weak solutions to the continuity equation}
\label{a.weak-continuity}

Before proceeding to the proof of Lemma~\ref{lem.properties-of-weak-solutions-cont}, we need to introduce the notion of weak Lebesgue points and show a version of Lebesgue differentiation for Banach space valued functions.

\begin{definition}
    Let $X$ a Banach space and let $f : (0,T) \to X$ in $L^1((0,T); X)$. We say that $t \in (0,T)$ is a \textit{weak Lebesgue point of $f$} if for every $\phi \in X'$, we have that $t$ is a Lebesgue point of $\phi\circ f :[0,T] \to \R$, i.e.\ 
    \[\lim_{\epsilon \to 0} \frac{1}{2\epsilon} \int_{t-\epsilon}^{t+\epsilon} \phi(f(s))\ ds = \phi(f(t)),\]
    or equivalently, we have that
    \[\lim_{\epsilon \to 0} \frac{1}{2\epsilon} \int_{t-\epsilon}^{t+\epsilon} f(s)\ ds = f(t),\]
    with convergence in the weak topology.
\end{definition}

\begin{lemma} 
    \label{lem.weak-leb-pt}
    Let $X$ a Banach space with a separable dual and let $f : [0,T] \to X$ in $L^1([0,T]; X)$. Then almost every $t \in [0,T]$ is a weak Lebesgue point of $f$.
\end{lemma}

\begin{proof}
    Let $\phi_j \in X'$ be a dense sequence in $X'$. Then, for each $j$, $\phi_j \circ f \in L^1((0,T), \R)$ so the set of it's Lebesgue points, $L_j$, is full measure. Additionally, note that $\|f\|_X \in L^1((0,T), \R),$ so it's set of Lebesgue points, $L_0$, is full measure. Let
    \[L := \bigcap_{j=0}^\infty L_j.\]
    Then $L$ is full measure also. We claim that for each $t \in L$, $t$ is a weak Lebesgue point of $f$. Fix $t \in L$ and $\phi \in X'$. Fix $\delta>0$ and let $\phi_j$ such that $\|\phi -\phi_j\| < \delta$. Then
    \begin{align*}
    \left \|\frac{1}{2\epsilon} \int_{t-\epsilon}^{t+\epsilon} \phi(f(s)) - \phi(f(t)) \right\| &\leq \left \|\frac{1}{2\epsilon} \int_{t-\epsilon}^{t+\epsilon} \phi_j(f(s))\ ds - \phi_j(f(t)) \right\| \\&\qquad+ \left\|(\phi - \phi_j) \left(\frac{1}{2\epsilon}\int_{t-\epsilon}^{t+\epsilon} f(s)\ ds - f(t)\right) \right\|.
    \end{align*}
    The first term goes to $0$ as $\epsilon \to 0$. For the second term we have that
    \begin{align*}
    \limsup_{\epsilon \to 0}\left\|(\phi - \phi_j) \left(\frac{1}{2\epsilon}\int_{t-\epsilon}^{t+\epsilon} f(s)\ ds - f(t)\right) \right\| &\leq \delta \left(\|f(t)\|_X + \limsup_{\epsilon \to 0} \frac{1}{2\epsilon} \int_{t-\epsilon}^{t+\epsilon} \|f(s)\|_X\right) 
    \\&= 2\delta \|f(t)\|_X.
    \end{align*}
    Thus, for any $\delta>0,$
    \[\limsup_{\epsilon \to 0} \left \|\frac{1}{2\epsilon} \int_{t-\epsilon}^{t+\epsilon} \phi(f(s)) - \phi(f(t)) \right\|  \leq 2 \delta \|f(t)\|_X.\]
    Taking $\delta \to 0$ then gives that
    \[\frac{1}{2\epsilon} \int_{t-\epsilon}^{t+\epsilon} \phi(f(s)) \to \phi(f(t)).\]
    Thus every $t \in L$ is a weak Lebesgue point of $f$, in particular the set of weak Lebesgue points of $f$ is full measure.
\end{proof}

\begin{proof}[Proof of Lemma~\ref{lem.properties-of-weak-solutions-cont}, Part 1]
    We are going to show the equality
    \[\theta = \theta_0 + \int_0^t \nabla \cdot (u\theta)\]
    pointwise for a.e.\ time. In particular, we will show it for every weak Lebesgue point of $\theta$, which are full measure by Lemma~\ref{lem.weak-leb-pt}. So fix $r \in (0,T)$ a weak Lebesgue point.

    We will check the equality distributionally
    \[\theta(r) = \theta_0 + \int_0^r \nabla \cdot (u \theta)(s)\ ds,\]
    so, letting $\phi \in C^\infty(\T^d)$ arbitrary, it suffices to verify
    \begin{align} \notag 0&=\int \theta(r,x) \phi(x) -\theta_0(x) \phi(x)-\int_0^r \nabla \cdot (u \theta)(s,x) \phi(x)\ dsdx
    \\&= \int \theta(r,x) \phi(x) -\theta_0(x) \phi(x)+\int_0^r u(s,x) \cdot \nabla \phi(x) \theta(s,x) \ dsdx.
    \label{eq.distributional-eq}
    \end{align}
    For each $\epsilon > 0$, let $\psi^\epsilon : [0,T] \to \R$ be defined as
    \begin{equation*}
    \psi^\epsilon(t) :=
    \begin{cases}
        1 & t \leq r-\epsilon,\\ 
        1 - \frac{t- (r-\epsilon)}{2\epsilon} & r-\epsilon \leq t \leq r+\epsilon, \\
        0 & t \geq r+\epsilon.
    \end{cases}
    \end{equation*}
    Then we test the equation for $\theta$ with $\psi^\epsilon(t) \phi(x)$ (technically this isn't smooth in time, but it an additional time mollification argument quickly shows this is not a problem), giving
    \begin{align*}
    0&= \int -\partial_t \psi^\epsilon \theta \phi + u \cdot \nabla \phi \theta \psi^\epsilon - \int \phi(x) \theta_0(x)
    \\&= \frac{1}{2\epsilon} \int_{r-\epsilon}^{r+\epsilon} \int \phi(x)  \theta(s,x)\ dxds - \int \theta_0 \phi\ dx+ \int \int_0^r u \cdot \nabla \phi \theta\ dsdr 
    \\&\qquad\qquad+ \int u \cdot \nabla \phi \theta (\psi^\epsilon(s) - 1_{[0,r]}(s))\ dsdr.
    \end{align*}
    Thus we see by comparing with~\eqref{eq.distributional-eq}, in order to conclude it suffices to show that, as $\epsilon \to 0$, the first term converges to $\int \phi(x) \theta(r,x)\ dx$ and then last term converges to $0$. 

    For the first term, the convergence is direct by the definition of a Lebesgue point. for the last term, we note that
    \[ \abs{\int u \cdot \nabla \phi \theta (\psi^\epsilon(s) - 1_{[0,r]}(s))\ dsdr}\leq \|u\|_{L^\infty} \|\nabla \phi\|_{L^\infty} \|\theta\|_{L^\infty_t L^1_x} \|\psi^\epsilon - 1_{[0,r]}\|_{L^1} \to 0.\]
\end{proof}

\begin{proof}[Proof of Lemma~\ref{lem.properties-of-weak-solutions-cont}, Part 2]
    By topological considerations, we have topological continuity if and only if we have sequential continuity, so we can just consider the problem of sequential continuity.
    
    Take the representation of $\theta$ given by by part (1). One can directly check that $\theta \in C([0,T], H^{-1})$ Let $t\in[0,T]$ and $t_n \in [0,T]$ such that $t_n \to t$. Let $t_{n_j}$ an arbitrary subsequence of $t_n$. Then $\theta(t_{n_j})$ is bounded in $L^2(\T^d)$ so is weakly compact. By taking a further subsequence, $s_j$, we have $\theta(s_j) \stackrel{L^2}{\weakto} \alpha$. But then $\theta(s_j) \stackrel{H^{-1}}\weakto \alpha$ also, but since $s_j \to t$ and by the continuity of $\theta$ into $H^{-1}$, we have $\theta(s_j) \stackrel{H^{-1}}\to \theta(t)$. Thus $\alpha = \theta(t)$ and so $\theta(s_j) \stackrel{L^2}{\weakto} \theta(t)$. Thus for every subsequence of $t_n$, there exists a further subsequence along which $\theta(s_j) \stackrel{L^2}{\weakto} \theta(t)$, we have that $\theta(t_j) \stackrel{L^2}{\weakto} \theta(t)$, so we conclude.
\end{proof}

\begin{proof}[Proof of Lemma~\ref{lem.properties-of-weak-solutions-cont}, Part 3]
     We prove the result for tensor products $\phi(t,x) = \psi(t) \gamma(x)$ and then we can conclude by the approximating an arbitrary $\phi$ by a linear combination of tensors.

    For fix $0 \leq s <r \leq T$ and a tensor $\phi(t,x) = \psi(t) \gamma(x)$. Let $\psi^\epsilon(t)$ be defined as follows. 
    \begin{equation*}
    \psi^\epsilon(t) := 
    \begin{cases}
        0 &  t\leq s,\\
        \psi(s+\epsilon)\frac{t-s}{\epsilon} & s \leq t \leq s+\epsilon,\\
        \psi(t) & s+\epsilon \leq t \leq r-\epsilon, \\
        \psi(r-\epsilon)\frac{r-t}{\epsilon} & r-\epsilon \leq t \leq r,\\
        0 & r \leq t.
    \end{cases}
    \end{equation*}
    By a simple mollification in time argument, one can verify we can test the equation for $\theta$ with $\psi^\epsilon(t) \gamma(x)$, giving
    \begin{align*}
    0&= \int -\partial_t \psi^\epsilon \theta \gamma + u \cdot \nabla \gamma \theta \psi^\epsilon\ dx dt
    \\&= \int \int_{s+\epsilon}^{r-\epsilon} -\partial_t \psi \theta \gamma + u \cdot \nabla \gamma \theta \psi + \int \int_{r-\epsilon}^r -\partial_t \psi^{\epsilon} \theta \gamma - \int \int_s^{s+\epsilon} \partial_t \psi^\epsilon \theta \gamma
    \\&\qquad + \int \int_{[s,s+\epsilon] \cup [r-\epsilon,r]} u \cdot \nabla \gamma \theta \psi^\epsilon.
    \end{align*}
    Thus to conclude, we need to show the following four limits
    \begin{align*}
    \int \int_{s+\epsilon}^{r-\epsilon} -\partial_t \psi \theta \gamma + u \cdot \nabla \gamma \theta \psi &\to \int \int_s^{r} -\partial_t \psi \theta \gamma + u \cdot \nabla \gamma \theta \psi
    \\ \int \int_{r-\epsilon}^r -\partial_t \psi^{\epsilon} \theta \gamma &\to \int \psi(r)\theta(r)\gamma\ dx\\
    \int \int_s^{s+\epsilon} \partial_t \psi^\epsilon \theta \gamma &\to \int \psi(s) \theta(s) \gamma\ dx\\
    \int \int_{[s,s+\epsilon] \cup [r-\epsilon,r]} u \cdot \nabla \gamma \theta \psi^\epsilon &\to 0.
    \end{align*}
    The first convergence is direct. The second and third are similar, so let's just do the second. Note that
    \[\int \int_{r-\epsilon}^r -\partial_t \psi^{\epsilon} \theta \gamma  = \psi(r-\epsilon)\frac{1}{\epsilon}\int_{r-\epsilon}^r \int \gamma(x)  \theta(t,x)\ dxdt.\]
    Then since $\theta \in C([0,T], L^2_w)$, we have that
    \[\alpha(t) := \int \gamma(x) \theta(t,x)\ dx \in C([0,T]),\]
    thus
    \[\frac{1}{\epsilon} \int_{r-\epsilon}^r \alpha(t) \to \alpha(r).\]
    Then by smoothness $\psi(r-\epsilon) \to \psi(r)$, so we get the limit for the product
    \[\int \int_{r-\epsilon}^r -\partial_t \psi^{\epsilon} \theta \gamma \to \psi(r) \alpha(r) = \psi(r) \int \gamma \theta(r)\ dx = \int \psi(r) \theta(r) \gamma\ dx.\]

    Lastly, the convergence of the fourth integral is direct also. So taking the all four integrals to their limits, we get the desired result.
\end{proof}

{\small
\bibliographystyle{alpha}
\bibliography{references,other-references}
}
 
\end{document}